\documentclass[11pt]{article}

\usepackage{amsmath}
\usepackage{amssymb}
\usepackage{amsthm}
\usepackage{mathtools}
\usepackage{graphicx}
\usepackage{caption}
\usepackage[sort,round]{natbib}

\usepackage{lmodern}

\usepackage[svgnames]{xcolor}
\usepackage[colorlinks=true, allcolors=DarkBlue]{hyperref}

\newtheorem{theorem}{Theorem}[section]

\newtheorem{corollary}[theorem]{Corollary}

\newtheorem{proposition}[theorem]{Proposition}
\newtheorem{remark}[theorem]{Remark}

\DeclareMathOperator{\TCE}{TCE}
\DeclareMathOperator{\TV}{TV}
\DeclareMathOperator{\TS}{TS}
\DeclareMathOperator{\TCov}{TCov}

\usepackage{anysize}
\marginsize{20mm}{20mm}{15mm}{30mm} 

\newcommand{\myfootnote}[1]{\linespread{1.0}\selectfont\footnote{#1}\linespread{1}\selectfont}

\begin{document}


\title{Wishart conditional tail risk measures: An analytic approach\footnote{We thank Katja Ignatieva for her useful suggestions and remarks. We also thank the participants at CICIRM 2024, Ningbo, China and IME 2024, Chicago, USA, for their remarks. The usual caveat applies.}
}

\author{Jos\'{e} Da Fonseca \thanks{Auckland University of Technology, Business School, Department of Finance, Private Bag 92006, 1142 Auckland, New Zealand. Phone: +64 9 9219999 ext. 5063. Email: jose.dafonseca@aut.ac.nz and PRISM Sorbonne EA 4101, Universit\'e Paris 1 Panth\'eon - Sorbonne, 17 rue de la Sorbonne, 75005 Paris, France. ORCID: 0000-0002-6882-4511} 
\and Patrick Wong\thanks{Monash University, Business School, Department of Econometrics and Business Statistics, Melbourne, VIC, 3800, Australia. Email: patrick.wong@monash.edu. ORCID: 0000-0002-6164-901X}
}

\date{\today}

\maketitle

\begin{abstract}
This study introduces a new analytical framework for quantifying multivariate risk measures. Using the Wishart process, which is a stochastic process with values in the space of positive definite matrices, we derive several conditional tail risk measures which, thanks to the remarkable analytical properties of the Wishart process, can be explicitly computed up to a one- or two-dimensional integration. These quantities can also be used to solve analytically a capital allocation problem based on conditional moments. Exploiting the stochastic differential equation property of the Wishart process, we show how an intertemporal {\color{black}(\textit{i.e.}, time-lagged)} view of these risk measures can be embedded in the proposed framework. Several numerical examples show that the framework is versatile and operational, thus providing a useful tool for risk management. 
\end{abstract}

\vspace{8cm}

\clearpage
\section{Introduction}

The modeling of risk measures is crucial for various applications such as risk management and reserving. The literature in this field is extensive and diverse. For example, \citet{LandsmanValdez2005} derive the tail conditional expectation for the exponential dispersion family, \citet{FurmanLandsman2005} examine the tail conditional expectation for a multivariate gamma portfolio with dependent risks, providing explicit formulas for tail conditional expectations and risk capital allocations. \citet{XuMao2013} solve a capital allocation problem based on a tail mean-variance model that requires the conditional higher moments up to order three of the underlying variables. They illustrate their model using a multivariate elliptical distribution. More recently, \citet{IgnatievaLandsman2015}, \citet{IgnatievaLandsman2019}, and \citet{IgnatievaLandsman2021} compute a closed-form expression for the tail conditional expectation for generalized hyperbolic/hyper-elliptical distributions, {\color{black}  while \citet{IgnatievaLandsman2025} extend these computations to the second order, that is, the tail variance.}  \citet{CheungPeraltaWoo2022} compute a wide variety of risk measures for multivariate matrix-exponential affine mixture models.\\

{\color{black}
In this work, we extend the perspective developed in \citet{DufresneGarridoMorales2009}, which computes the Tail Conditional Expectation (TCE) for a univariate loss, to a framework that enables the computation of the TCE for risks, whether individual risks or aggregates, of any moment, across a wide range of random variable distributions, provided their moment generating function (MGF) is available.\footnote{Whenever the MGF is not available, as in the case of the Pareto distribution, for example, the characteristic function can be used instead; see \citet{DufresneGarridoMorales2009} for the scalar case.}\\

Additionally, our framework incorporates general dependence among risks, or losses across individual lines, by considering distributions on the space of symmetric positive definite matrices, for which the joint MGF is known for many distributions. This approach does not rely on a specific distribution to derive the TCE or any higher order tail conditional moments. If the static distribution used to compute risk measures can be embedded into a (time-dependent) stochastic process, it becomes possible to compute intertemporal or time-lagged TCEs (or any higher order tail conditional moments). The intertemporal/time-lagged TCE represents the expected loss at time \( t_1 \), given that the loss at time \( t_0 \) exceeded a certain threshold \( x_* \), where \( t_0 < t_1 \). We can also compute the TCE at two different time points for dependent losses, provided the joint MGF is known at both \( t_0 \) and \( t_1 \) for the individual lines. This property extends the static or time-independent perspective analyzed in the aforementioned literature to the dynamic or time-series point of view considered in more recent studies (e.g., \citet{ChavezDemoulinGuillou2018}, \citet{PattonZiegelChen2019}, \citet{GoegebeurGuillou2024}, \citet{GoegebeurGuillou2025}, or \citet{LiuLiuZhao2025}).\myfootnote{{\color{black} The dynamic property can also be understood in the sense of \citet{Riedel2004}, \citet{ArtznerDelbaenEberHeathKu2007}, \citet{HardyWirch2004}, or \citet{DevolderLebegue2017}, which differs from the purely time-series point of view adopted here.}} Compared to this literature, our modeling strategy based on the Wishart process provides control over both the evolution of losses and their dependencies.\\
}

In modeling portfolio losses, it is important to recognize that losses can exhibit dependence across different losses. Traditional copula-based methods often impose constraints on the types of dependencies we can effectively model, limiting their flexibility. {\color{black} The literature on risk measures using copulas to handle dependence is extensive. Without pretending to be exhaustive, let us mention \citet{CousinDiBernardino2014} and \citet{DiBernardino2015}.} {\color{black} An alternative to the copula approach is the factor model approach of \citet{FurmanLandsman2005}, which is based on a vector of independent positive variables. In this framework, all losses are expressed as linear functions of this vector and share a common component, akin to a principal component analysis. This modeling strategy imposes certain constraints on the sign of the dependencies.} The alternative that we propose is to capture these relationships by directly specifying both the losses and their dependence structure simultaneously. A natural framework for this purpose is to view losses within the space of positive definite matrices, similar to how covariances are treated in portfolio theory. This approach not only aligns with the inherent structure of dependence but also allows for richer modeling of the interactions between different losses.\\

The Wishart process, originally defined by \cite{Bru1991}, is a flexible process applied in option pricing, mortality modeling, and equity modeling. It is a positive definite matrix process, making it suitable for modeling positive events such as losses and their dependencies. {\color{black} As such, it can be seen as an extension of  \citet{FurmanLandsman2005} where the dependence between the losses is directly modeled.} As an affine process, its MGF is known in closed form. Furthermore, the derivative of the MGF can be computed explicitly, a property that plays a crucial role in operationalising the implementation of risk measures. One of the key ingredients of our approach is to express different risk measures in terms of the MGF of the losses, thereby enabling us to fully exploit this property of the Wishart process. Note that many of our results hold for any distribution, as long as its MGF is known. We highlight our framework using the Wishart process, demonstrating that it satisfies the integrability conditions required by our framework and providing a numerical implementation to showcase the speed and accuracy of our results.\\

The paper is organized as follows. In Section~\ref{sec:Framework}, we introduce the Wishart process and provide its key results, including its MGF and the joint MGF at two different dates. In Section~\ref{sec:ConditionaTail}, we present our main result: the ability to compute multivariate risk measures using their Fourier transforms. By utilizing the Fourier transform, we express conditional tail expectations in terms of the joint MGF of random variables and derive these expressions even when a temporal element is introduced. We then prove that the Wishart process meets the integrability conditions for applying these results. We show that our approach is comparable to existing frameworks in terms of numerical cost. We also demonstrate how well-known conditional risk measures can be expressed in our framework. In Section~\ref{sec:numerical_implementation}, we provide a numerical implementation of the framework on the Wishart process, demonstrating accurate estimates and the effect of dependence on the measures. Finally, we conclude in Section~\ref{Conclusion}.

\section{The mathematical framework}\label{sec:Framework}

Given a filtered probability space \((\Omega,\mathcal{F}, \mathcal{F}_t, \mathbb{P})\), we denote by \(\mathbb{E}\left[\,\cdot\,\right]\) the expectation under the probability measure \(\mathbb{P}\), and by \(\mathbb{E}_t\left[\,\cdot\,\right]:=\mathbb{E}\left[\,\cdot\,|\mathcal{F}_t \right]\) the conditional expectation under \(\mathbb{P}\). The Wishart process, proposed in \citet{Bru1991}, satisfies the matrix stochastic differential equation
\begin{equation}
dx_t = (\omega + m x_t + x_t m^\top )dt + \sqrt{x_t}dw_t \sigma + \sigma^\top dw_t^\top \sqrt{x_t}\,, \label{eq:Wishart}
\end{equation}
where \(x_t\) is an \(n\times n\) matrix that belongs to the set of positive definite matrices denoted \(\mathbb{S}_n^{++}\) (\(\mathbb{S}_n^{+}\) stands for the set of positive semi-definite matrices, \(\mathbb{S}_n\) the set of real \(n\times n\) symmetric matrices). The matrix \(m\) belongs to the set of \(n\times n\) real matrices denoted \(\mathsf{M}(n)\) and is such that \(\lbrace \Re(\lambda_i^{m})<0; i=1,\ldots,n \rbrace\), where \(\lambda_i^{m}\in Spec(m)\) for \(i=1,\ldots,n\) and \(Spec(m)\) is the spectrum of the matrix \(m\), while \(\Re(\,\cdot\,)\) stands for the real part. Lastly, \(\lbrace w_t;t\geq 0 \rbrace\) is a standard matrix Brownian motion of dimension \(n\times n\) (\textit{i.e.}, a matrix of \(n^2\) independent scalar standard Brownian motions) under the probability measure \(\mathbb{P}\), and \(\cdot^\top\) stands for the matrix transposition. Thanks to the invariance of the law of the Brownian motion to rotations and the polar decomposition of \(\sigma\), we can restrict to \(\sigma \in \mathbb{S}_n^{++}\) without loss of generality. The matrix \(\omega \in \mathbb{S}_n^{++}\) satisfies certain constraints, clarified below, involving \(\sigma^2\) to ensure the positive definiteness of the matrix process \(x_t\). The quantity \(\sqrt{x_t}\) is well-defined and unique since \(x_t \in \mathbb{S}_n^{++}\). We denote by \(e_{ij}\) the basis of \(\mathsf{M}(n)\); it is the \(n\times n\) matrix with \(1\) in the \((i,j)\) place and \(0\) elsewhere. The identity matrix of \(\mathsf{M}(n)\) is denoted \(I_{n}\). We also introduce the symbols \(\textup{tr}[\,\cdot\,]\) for the trace of a matrix, \(\textup{vec}\left(\cdot\right)\) for the operator which transforms an \(n \times n\) matrix into an \(n^2\) vector by stacking the columns, \(\otimes\) for the Kronecker product, \(\det(\cdot)\) for the determinant of a matrix, and \(\|\cdot\|\) for the matrix norm associated with \(\textup{tr}[\,\cdot\,]\).\\

The Wishart process was initially defined and analyzed in \citet{Bru1991} under the assumption that $\omega=\beta \sigma^2$ with $\beta \in \mathbb{R}_+$ such that $\beta \geq n+1$  to ensure that for all $t>0$ we have $x_t \in \mathbb{S}_n^{++}$. Hereafter, this specification will be referred to as the Bru case. It was later extended in \citet{MayerhoferPfaffelStelzer2011} (see also \citet{CuchieroFilipovicMayerhoferTeichmann2011}) to the case $\omega \in \mathbb{S}_n^{++}$ and proved that if 
\begin{align}
	\omega \succeq \beta \sigma^2\,, \label{eq:PositiveConstraint}
\end{align}
with $\beta\geq n+1$, where \eqref{eq:PositiveConstraint} means that $\omega -\beta \sigma^2 \in \mathbb{S}_n^{++}$,  then for all $t>0$ we have $x_t\in \mathbb{S}_n^{++}$.\\

The infinitesimal generator of the Wishart process is given by \citet{Bru1991}:
\begin{align}
\mathcal{G} = \textup{tr}[(\omega + m x + xm^\top)D + 2 x D \sigma^2 D]\,, \label{eq:InfinitesimalGenerator}
\end{align}
where $D$ is the $(n\times n)$ matrix operator $D_{ij}:=\partial_{x_{ij}}$.\\

\citet{Bru1991} showed that as the Wishart process is affine, its MGF is exponentially affine. More precisely, the MGF of $x_t$ is given by the following proposition, see \citet{GrasselliTebaldi2008} and \citet{Alfonsi2015} for details.

\begin{proposition}\label{prop:MGFWishart}
Let $(x_t)_{t\geq 0}$ be a Wishart process given by ~\eqref{eq:Wishart}, and denote the MGF of $x_t$ by
\begin{equation}
\Phi(t,\theta_0, x_0):=\mathbb{E}\left[ \exp\left( \textup{tr}[\theta_0 x_t] \right)\right]\,, \label{eq:mgf}
\end{equation}
where $\theta_0 \in \mathbb{S}_n$. Then 
\begin{equation}
\Phi(t,\theta_0,x_0)= \exp\left(  \textup{tr} [a(t,\theta_0)  x_0] +  b(t,\theta_0) \right)\,, \label{eq:solPhi}
\end{equation}
with the deterministic functions $(a(t,\theta_0),b(t,\theta_0))$, where $a(t,\theta_0)\in \mathsf{M}(n)$ and $b(t,\theta_0) \in \mathbb{R}$, satisfying the system 
\begin{align}
a' &=a m + m^\top a + 2a \sigma^2 a\,, \label{eq:Riccati}\\
b' &= \textup{tr} [\omega a]\,, \label{eq:odeb}
\end{align}
with initial conditions $a(0,\theta_0)=\theta_0$ and $b(0,\theta_0)=0$. As usual $\cdot^\prime$ denotes the time derivative. The matrix Riccati ordinary differential equation (ODE) \eqref{eq:Riccati} admits the solution 
	\begin{align}
		a(t,\theta_0)=e^{t m^\top }(I_n -2 \theta_0\varsigma_t)^{-1}\theta_0 e^{t m}\,, \label{eq:Asimple}
	\end{align}
	with $ \varsigma_t:=\int_0^t e^{(t-s)m}\sigma^2e^{(t-s)m^\top}ds\in \mathbb{S}_n^{++}$, which is given by $\textup{vec}(\varsigma_t)=\mathsf{A}^{-1}(e^{t\mathsf{A}} - I_{n^2})\textup{vec}(\sigma^2)$ with  $\mathsf{A} = \left(I_{n}\otimes m+m\otimes I_{n}\right)$,  while $ b(t,\theta_0)$ is obtained by a numerical integration of \eqref{eq:odeb}.\\
	
	If we further assume that $\omega=\beta \sigma^2$ with $\beta \in \mathbb{R}$ and $\beta \geq n+1$ then $b(t,\theta_0)$ is given by
\begin{align}
e^{b(t)}&=\frac{1}{\det(I_n -2\varsigma_t\theta_0)^{\beta/2}}\,. \label{eq:ExpB}
\end{align}
\end{proposition}

Note that for $\Phi(t,\theta_0, x_0)$ to be well-defined, the matrix $(I_n -2 \theta_0\varsigma_t)$ in \eqref{eq:Asimple} should be invertible and it is possible as long as \(\|\theta_0\|\) is sufficiently small since \(\varsigma_t\in \mathbb{S}_n^{++}\).\\
 
From the equations \eqref{eq:Asimple} and \eqref{eq:ExpB}, it is straightforward to prove the following result. 
\begin{corollary}\label{coro:MGFInfinity} 
If we assume that $\omega=\beta \sigma^2$ with $\beta \in \mathbb{R}$ and $\beta \geq n+1$, then the MGF of $x_{\infty}=\lim_{t \to +\infty} x_t$, denoted $\Phi^{\infty}(\theta):=\lim_{t\to +\infty}\Phi(t,\theta, x_0)=\mathbb{E}\left[e^{\textup{tr}[\theta x_{\infty}]}\right]$, is given for $\theta \in \mathbb{S}_n$ by 
\begin{align} 
\Phi^{\infty}(\theta)=\frac{1}{\det(I_n -2\varsigma_\infty\theta)^{\beta/2}}, \label{eq:MGFInfinity} 
\end{align} 
with $\varsigma_\infty = \lim_{t\to +\infty}\varsigma_t$ that is well-defined. 
\end{corollary}

Note that \eqref{eq:MGFInfinity} is the MGF of a matrix gamma distribution, and for diagonal parameters, the vector $(x_{11,\infty},\ldots, x_{nn,\infty})$ is an $n$-dimensional vector of independent scalar gamma variables such that $x_{ii,\infty}$ has a scale $2(\varsigma_{\infty}){ii}$ and shape $\beta/2$. This type of distribution is used in \citet{FurmanLandsman2005} to analyse the TCE risk measure. {\color{black} Since the components of the vector are independent, the authors express the losses as linear functions of this vector, with the dependence between losses captured by a common component of the vector gamma distribution. This common-factor perspective imposes a sign constraint on the dependencies between losses.} Note that in the non-diagonal case, the components $x_{ii,\infty}$ are dependent, and therefore it provides an alternative modeling strategy to create dependence compared to the one used in \citet{FurmanLandsman2005}, which consists of specifying a common factor.\\

To compute the conditional risk measures, we will need the derivative of the MGF with respect to a scalar parameter that multiplies its matrix argument. This relies on \citet[Theorem 4, p. 127]{Lax2007} and \citet[Theorem 2, p. 124]{Lax2007}.

\begin{proposition}\label{prop:MGFWishartDerivative} 
Given a Wishart process $(x_t)_{t\geq 0}$ and for $\epsilon>0$ let $\nu \in \mathbb{R}$ with $\nu \in ]-\epsilon,\; \epsilon[$. Assume that $ \theta_0, \theta_1\in \mathbb{S}_n$ (that do not depend on $\nu$), and define the MGF
	\begin{align}
		\bar{\Phi}(t,\nu,\theta_0,\theta_1,x_0)	&:= \mathbb{E}\left[\exp\left(\textup{tr}[(\theta_0+\nu \theta_1)x_t]\right) \right]\nonumber \\
		&= \Phi(t,\theta_0+\nu\theta_1,x_0)\,, \label{eq:PhiNu}
	\end{align}
	with $\Phi(.,.,.)$ given by \eqref{eq:solPhi}. Then we have
	\begin{align}
		\bar{\Phi}_{\nu}(t,\theta_0,\theta_1 ,x_0) 	&:=\partial_{\nu} \bar{\Phi}(t,\nu,\theta_0,\theta_1,x_0)|_{\nu = 0}\nonumber \\
																								&=\left(\textup{tr}[a_{\nu}(t,\theta_0,\theta_1) x_0] + b_{\nu}(t,\theta_0,\theta_1) \right)\Phi(t,\theta_0,x_0)\,,\label{eq:MGFWishartDerivative}
	\end{align}
	with 
	\begin{align}
		a_{\nu}(t,\theta_0,\theta_1) 	&= e^{ t m^\top }(I_n -2 \theta_0\varsigma_t)^{-1}2\theta_1\varsigma_t(I_n -2 \theta_0\varsigma_t)^{-1}\theta_0e^{tm} + e^{t m^\top}(I_n -2 \theta_0\varsigma_t)^{-1}\theta_1e^{t m}\,,\label{eq:Anu}
	\end{align}
	while $ b_{\nu}(t,\theta_0,\theta_1)=\int_0^t \mathsf{tr}[\omega a_{\nu}(u,\theta_0,\theta_1) ] du$.\\
	
	If we further assume that $\omega=\beta \sigma^2$ with $\beta \in \mathbb{R}$ and $\beta \geq n+1$ then  $b_\nu(t,\theta_0,\theta_1)= \beta\textup{tr}[(I_n-2\varsigma_t\theta_0)^{-1}\varsigma_t \theta_1]$.
	
\end{proposition}

Note that for $\bar{\Phi}(t,\nu,\theta_0,\theta_1,x_0)$ to be well-defined, the norm of $\theta_0 + \nu \theta_1$ needs to be sufficiently small. If $\epsilon$ is sufficiently small and the norm of $\theta_0$ is also sufficiently small, then the function is well-defined according to Proposition~\ref{prop:MGFWishart}. Note also that if the norm of $\theta_0$ is sufficiently small, then $(I_n - 2 \theta_0 \varsigma_t)$ is invertible since $\varsigma_t \in \mathbb{S}_n^{++}$.\\

The above proposition gives the first order derivative of the moment generating function with respect to \(\nu\). Below, we will need the \textcolor{blue}{\(q\text{-th}\) order derivative, where \(q \in  \mathbb{N}\),} of this function with respect to this parameter.

\begin{corollary}\label{coro:MGFWishartQDerivative} 
Assume that $\omega=\beta \sigma^2$ with $\beta \in \mathbb{R}$ and $\beta \geq n+1$ and let $\bar{\Phi}(.,.,.,.)$ be given by \eqref{eq:PhiNu}. Then the $q\text{-th}$ order derivative of $\bar{\Phi}(t,\nu,\theta_0,\theta_1,x_0)$ with respect to $\nu$ (evaluated for $\nu=0$) is given by
	\begin{align}
		\bar{\Phi}_{\nu^q}(t,\theta_0,\theta_1 ,x_0) &:=\left.\frac{\partial^q \bar{\Phi}(t,\nu,\theta_0,\theta_1,x_0)}{\partial \nu^q}\right|_{\nu = 0}.\label{eq:MGFWishartQDerivative}
	\end{align}
\end{corollary}


The previous results can be extended to the MGF of the process at two different times thanks to the affine property of the Wishart process. Indeed, in addition to the MGF of the Wishart process at a given point in time, \textit{i.e.}, the marginal point of the process, it is of practical interest to consider two points in time $0<t_0<t_1$ and the MGF of $(x_{t_0},x_{t_1})$. This will allow us to consider losses at two different dates, thus introducing an intertemporal {\color{black} or time-lagged} perspective to the risk measurement problem. It is easy to obtain the following result from Proposition~\ref{prop:MGFWishart} and the use of the law of iterated expectation.

\begin{proposition}\label{prop:MGFWishart2Dates}
Let $(x_t)_{t\geq 0}$ a Wishart process given by ~\eqref{eq:Wishart}, and assume two dates $0<t_0<t_1$. Then the MGF of $(x_{t_0},x_{t_1})$ is 
\begin{align}
\Phi(t_0,t_1,\theta_0,\theta_1, x_0)&:=\mathbb{E}\left[ \exp\left( \textup{tr}[\theta_0 x_{t_0}] + \textup{tr}[\theta_1 x_{t_1}] \right)\right] \nonumber\\
&= \exp\left( \textup{tr} [a(t_0,\theta_0 + a(t_1-t_0,\theta_1))  x_0] +  b(t_0,\theta_0) +  b(t_1-t_0,\theta_1)\right) \label{eq:MGF2}\,, 
\end{align}
where $\theta_0$ and $\theta_1$ belong to $\mathbb{S}_n$, $a(.,.)$ and $b(.,.)$ are given by Proposition~\ref{prop:MGFWishart}.
\end{proposition}

\begin{proof}
	Using the law of iterated expectation, we have
	\begin{align*}
		\Phi(t_0,t_1,\theta_0,\theta_1, x_0) &= \mathbb{E}\left[ \exp\left( \textup{tr}[\theta_0 x_{t_0}] + \textup{tr}[\theta_1 x_{t_1}] \right)\right]\\
		&= \mathbb{E}\big[ \mathbb{E}[ \exp\left( \textup{tr}[\theta_0 x_{t_0}] + \textup{tr}[\theta_1 x_{t_1}] \right)  |\mathcal{F}_{t_0}]\big]\\
		&= \mathbb{E}\left[\exp \left(\textup{tr}(\theta_0x_{t_0})\right) \cdot \exp(\textup{tr}[a(t_1 - t_0, \theta_1)x_{t_0}] + b(t_{1} - t_0, \theta_1)  )
		\right]\\
		&= \exp\left( \textup{tr} [a(t_0,\theta_0 + a(t_1-t_0,\theta_1))  x_0] +  b(t_0,\theta_0) +  b(t_1-t_0,\theta_1)\right).
		\end{align*}
\end{proof}


A natural extension of Proposition~\ref{prop:MGFWishart2Dates} that we will need is the following MGF.

\begin{proposition}\label{prop:MGF2WishartDerivative2} 
	Given a Wishart process $(x_t)_{t\geq 0}$, $\nu_0, \nu_1 \in \mathbb{R}$ and $\bar{\theta}_0, \theta_0, \theta_1\in \mathbb{S}_n$ (that do not depend on $\nu_0$, $\nu_1$). Define the MGF
	\begin{align}
		\tilde{\Phi}(t_0,t_1,\bar{\theta}_0,\nu_0,\nu_1,\theta_0,\theta_1,x_0)	&:= \mathbb{E}\left[\exp\left(\textup{tr}[\bar{\theta}_0x_{t_0}] + \textup{tr}[(\nu_0\theta_0+\nu_1 \theta_1)x_{t_1}]\right) \right] \nonumber\\
		&= \Phi(t_0,t_1,\bar{\theta}_0,(\nu_0\theta_0+\nu_1 \theta_1),x_0)\,, \label{eq:barPhiTime2}
	\end{align}
	with $\Phi(.,.,.,.,.)$ as in \eqref{eq:MGF2}. Then its derivative with respect to \(\nu_0\) and \(\nu_1\) is 
	
	\begin{align}
		\tilde{\Phi}_{\nu_0\nu_1}&(t_0,t_1,\bar{\theta}_0,\theta_0,\theta_1,x_0) :=\partial_{\nu_0\nu_1}^2 \tilde{\Phi}(t_0,t_1,\bar{\theta}_0,\theta_0,\nu_0,\nu_1,\theta_1,x_0)|_{\nu_0 = 0,\nu_1 = 0}\label{eq:MGF2WishartDerivativeTime2} \\
	&=\left(\textup{tr}[\partial_{\nu_0\nu_1}^2a(t_0,\bar{\theta}_0 +a(t_1-t_0,\nu_0\theta_0 + \nu_1\theta_1)) x_0] + \partial_{\nu_0\nu_1}^2b(t_1-t_0,\nu_0\theta_0 + \nu_1\theta_1) \right)|_{\nu_0 = 0,\nu_1 = 0} \Phi(t_0,\bar{\theta}_0,x_0)\nonumber \\
	&+\left(\textup{tr}[\partial_{\nu_0}a(t_0, \bar{\theta}_0 + a(t_1-t_0,\nu_0\theta_0 + \nu_1\theta_1)) x_0] + \partial_{\nu_0}b(t_1-t_0,\nu_0\theta_0 + \nu_1\theta_1) \right) \nonumber \\
	&\times \left(\textup{tr}[\partial_{\nu_1}a(t_0, \bar{\theta}_0 +a(t_1-t_0,\nu_0\theta_0 + \nu_1\theta_1)) x_0] + \partial_{\nu_1}b(t_1-t_0,\nu_0\theta_0 + \nu_1\theta_1) \right)|_{\nu_0 = 0,\nu_1 = 0}\Phi(t_0,\bar{\theta}_0,x_0) \nonumber
	\end{align}
	
that involves $\partial_{\nu_0}a(t_0,\bar{\theta}_0 + a(t_1-t_0,\nu_0\theta_0+  \nu_1\theta_1))$, $\partial_{\nu_1}a(t_0,\bar{\theta}_0 + a(t_1-t_0,\nu_0\theta_0+  \nu_1\theta_1))$ and $\partial_{\nu_0\nu_1}^2a(t_0,\bar{\theta}_0 + a(t_1-t_0,\nu_0\theta_0+  \nu_1\theta_1))$, which have to be evaluated for $\nu_0=0$ and $\nu_1=0$, with $a(.,.)$ given by \eqref{eq:Asimple}. These functions can be explicitly computed.\\

This derivative also involves the functions $\partial_{\nu_0}b(t_1-t_0,\nu_0\theta_0 + \nu_1\theta_1)$, $ \partial_{\nu_1}b(t_1-t_0,\nu_0\theta_0 + \nu_1\theta_1)$ and $\partial_{\nu_0\nu_1}^2b(t_1-t_0,\nu_0\theta_0 + \nu_1\theta_1)$. If we further assume that $\omega=\beta \sigma^2$ with $\beta \in \mathbb{R}$ and $\beta \geq n+1$ then they are equal for $\nu_0=0$ and $\nu_1=0$ to $\frac{\beta}{2}\textup{tr}[2\varsigma_{t_1-t_0} \theta_0]$, $\frac{\beta}{2}\textup{tr}[2\varsigma_{t_1-t_0} \theta_1]$ and $2\beta\textup{tr}[\varsigma_{t_1-t_0}\theta_1\varsigma_{t_1-t_0} \theta_0]$, respectively.
\end{proposition}

The previous derivative function can be extended to a higher order by performing additional derivations. The following corollary is more for notational purposes.

\begin{corollary}\label{coro:MGF2WishartQDerivativeTime2} 
Assume that $\omega=\beta \sigma^2$ with $\beta \in \mathbb{R}$ and $\beta \geq n+1$ and let $\tilde{\Phi}(.,.,.,.)$ be given by \eqref{eq:barPhiTime2}. Then the $(p\text{-th},q\text{-th})$ order derivative, with $(p,q)\in \mathbb{N}^2$, of $\tilde{\Phi}(t_0,t_1,\bar{\theta}_0,\nu_0,\nu_1,\theta_0,\theta_1,x_0)$ with respect to $(\nu_0,\nu_1)$ (evaluated for $\nu_0=0$ and $\nu_1=0$) is 
	\begin{align}
		\tilde{\Phi}_{\nu_0^p\nu_1^q}(t_0,t_1,\bar{\theta}_0,\theta_0,\theta_1 ,x_0) &:=\left.\frac{\partial^{p+q}\tilde{\Phi}(t_0,t_1,\bar{\theta}_0,\nu_0,\nu_1,\theta_0,\theta_1,x_0)}{\partial \nu_0^p\partial \nu_1^q}\right|_{\nu_0 = 0,\nu_1=0}.\label{eq:MGF2WishartQDerivativeTime2} 
	\end{align}
\end{corollary}

\section{Conditional tail risk measures}\label{sec:ConditionaTail}
In this section, we begin by reformulating the key quantities required to compute conditional higher moments in terms of their Fourier transform representations, as outlined in Section~\ref{sec:MultiFourierTransform}. While these Fourier transform representations of the  tail conditional expectations may seem standard, to the best of our knowledge, they have not been previously reported in the literature {\color{black}in the multivariate and/or the higher order moments cases.}\myfootnote{{\color{black} A Fourier transform representation of the \(\TCE\) in the univariate case appears in \citet{DufresneGarridoMorales2009}.}} Even if they are of interest, these Fourier transform representations involve multidimensional integrations that are numerically demanding. As such, they do not provide an interesting alternative to conditional higher moments computed using the density of the state variable, which is the standard approach used in the literature. To overcome this difficulty, in Section~\ref{sec:OneFourierTransform} we introduce one-dimensional Fourier transform representations for these risk measures, thereby significantly reducing the computational cost. However, this comes at the cost of computing the (higher) derivative of the MGF of the Wishart process, but this can be done as shown in the previous section. Note that these results apply not only to the Wishart process, but to any distribution with a known MGF,\myfootnote{{\color{black}Throughout, we treat the MGF as a function of a complex parameter and work on the domain where it is finite. On the imaginary axis $z=\mathrm{i}v$ this coincides with the characteristic function $\varphi_Y(v)=\mathbb{E}\!\left[e^{\mathrm{i}vY}\right]$, which always exists. Unless stated otherwise, ``MGF'' refers to this complex-argument version.}} so these results have broad applicability that goes well beyond the particular case of the Wishart process that we use here. We begin by presenting the Fourier transform representations of the risk measures in a very general way in Section~\ref{sec:MultiFourierTransform} and Section~\ref{sec:OneFourierTransform}, and then illustrate them in the particular case of the Wishart process in Section~\ref{sec:WishartRiskMeasures}. We end this section by providing a comparison, in terms of computational cost, of our approach with the standard strategy used in the literature.

\subsection{Multidimensional Fourier transform representations for risk measures}\label{sec:MultiFourierTransform}

When considering only one loss or a scalar loss, then the standard risk measures are the higher moments conditional on the loss being above a certain threshold. By Bayes' theorem we have the following expression for the moment of a variable $Y$ conditional on the fact that this variable is above a certain threshold:
\begin{align}
\mathbb{E}[Y^p | Y>y_*]=\frac{\mathbb{E}[Y^p \mathbf{1}_{\lbrace Y>y_* \rbrace}]}{\mathbb{E}[\mathbf{1}_{\lbrace Y>y_*\rbrace}]}, \label{eq:ConditionalMomentYp}
\end{align}
which for $p=1$ is commonly denoted in the literature as $\TCE_Y(y_*)$ and is the tail conditional expectation. Thus, the evaluation of \eqref{eq:ConditionalMomentYp} amounts to the calculation of $\mathbb{E}[Y^p \mathbf{1}_{\lbrace Y>y_* \rbrace}]$, which is commonly done using the density of the variable \(Y\). We reformulate this expectation in terms of the MGF of this variable, thus avoiding the use of its density. The proposition below shows how to calculate this.

\begin{proposition}\label{prop:ExpectedShortfallFourierTransform}
Denote $Y$ a positive (scalar) random variable and denote by $\Phi_Y(v)=\mathbb{E}\left[e^{vY}\right]$ for $v\in \mathbb{R}$ its MGF. Let $y_*\in \mathbb{R}_+$,  $\alpha_1\in \mathbb{R}_+$ and assume $u_1 \in \mathbb{R} \to \frac{\Phi_Y(\alpha_1-\mathrm{i}u_1)}{(\alpha_1-\mathrm{i}u_1)^{j+1}}$ (with $\mathrm{i}=\sqrt{-1}$) is integrable for $j=0,\ldots, p$ with $p \in \mathbb{N}$. Then 
\begin{align}
\mathbb{E}[Y^p \mathbf{1}_{\lbrace Y>y_* \rbrace}]=\int_{0}^{+\infty}\sum_{j=0}^p \binom{p}{j}  \frac{y_*^{p-j} c_j}{\pi }\Re\left( \frac{e^{-(\alpha_1 -\mathrm{i}u_1)y_* }\Phi_Y(\alpha_1-\mathrm{i} u_1)}{(\alpha_1-\mathrm{i}u_1)^{j+1}}\right)du_1,\label{eq:MomentYp}
\end{align}
with \(c_j=\Gamma(j+1)\) and $\Gamma(.)$ the gamma function.
\end{proposition}
\begin{proof}
We have 
\begin{align*}
\mathbb{E}[Y^p \mathbf{1}_{\lbrace Y>y_* \rbrace}]&=\sum_{j=0}^p \binom{p}{j} \mathbb{E}[(Y-y_*)_+^j]y_*^{p-j},
\end{align*}
and define $V=Y-y_*$ so that $\mathbb{E}[(Y-y_*)_+^j]=\mathbb{E}[(V)_+^j]$. \citet[3.2(3)]{BatemanEderlyi1954} {\color{black} is
\begin{align}
\frac{2\pi (v)_+^{j+1-1}e^{-\alpha_1 v}}{\Gamma(j+1)} = \int_{-\infty}^{+\infty}\frac{e^{-\mathrm{i} vu_1}}{(\alpha_1-\mathrm{i}u_1)^{j+1}}du_1,
\end{align}
with \(\alpha_1>0\). After reorganizing the terms, multiplying by the density of \(V\) that we denote \(g(v)\), integrating with respect to the variable \(v\) we reach
\begin{align}
\int_{-\infty}^{+\infty}(v)_+^{j} g(v) dv = \frac{\Gamma(j+1)}{2\pi} \int_{-\infty}^{+\infty}\int_{-\infty}^{+\infty}\frac{e^{\alpha_1 v}e^{-\mathrm{i} vu_1}}{(\alpha_1-\mathrm{i}u_1)^{j+1}}g(v)du_1dv,
\end{align}
and the above equality rewrites as
}
\begin{align*}
\mathbb{E}[(V)_+^j]&=\int_{-\infty}^{+\infty}\frac{\Gamma(j+1)}{2\pi }\int_{-\infty}^{+\infty}\frac{e^{\alpha_1v}e^{-\mathrm{i} u_1v}}{(\alpha_1-\mathrm{i} u_1)^{j+1}}du_1g(v) dv\\
&=\int_{-\infty}^{+\infty}\frac{\Gamma(j+1)}{2\pi }\frac{1}{(\alpha_1-\mathrm{i} u_1)^{j+1}} \int_{-\infty}^{+\infty}e^{(\alpha_1 -\mathrm{i} u_1)v}g(v)dvdu_1\\
&=\int_{-\infty}^{+\infty}\frac{\Gamma(j+1)}{2\pi }\frac{\mathbb{E}[e^{(\alpha_1 -\mathrm{i} u_1)V}] }{(\alpha_1-\mathrm{i} u_1)^{j+1}} du_1,
\end{align*}
{\color{black} after applying Fubini's theorem}. Since $\mathbb{E}[e^{(\alpha_1 -\mathrm{i} u_1)V}]=\mathbb{E}[e^{(\alpha_1 -\mathrm{i} u_1)Y}]e^{-(\alpha_1 -\mathrm{i} u_1)y_*}$ we obtain
\begin{align*}
\mathbb{E}[(Y-y_*)_+^j]=\frac{\Gamma(j+1)}{2\pi }\int_{-\infty}^{+\infty}\frac{e^{-(\alpha_1 -\mathrm{i} u_1)y_*} \Phi_Y(\alpha_1-\mathrm{i}u_1)}{(\alpha_1-\mathrm{i} u_1)^{j+1}}du_1,
\end{align*}
that leads to the result after considering the real part of the right hand side of the above equation since the left hand side is real. Note that to apply Fubini's theorem, it requires the integrability of $ u_1 \in \mathbb{R} \to \frac{\Phi_Y(\alpha_1-\mathrm{i} u_1)}{(\alpha_1-\mathrm{i} u_1)^{j+1}}$, {\color{black} which depends on the MGF of the variable \(Y\).}
\end{proof}
{\color{black} Note that the result applies to distributions defined not only on \(\mathbb{R}_+\), but also on \(\mathbb{R}\), provided the integrability condition is satisfied, and if we opted to use the characteristic function $\varphi_Y(v) = \mathbb{E}\left[e^{\mathrm{i}vY}\right]$ (which always exists for real $v$) instead of the MGF $\Phi_Y(v) = \mathbb{E}\left[e^{vY}\right]$, it would amount to changing $\Phi_Y(\alpha_1-\mathrm{i} u_1)$ in \eqref{eq:MomentYp} to $\varphi_Y(-\mathrm{i}\alpha_1-u_1)$. The proof remains unchanged. Since we use the MGF of the Wishart process (see, for example, \citet{GrasselliTebaldi2008} and \citet{Alfonsi2015}), we formulate our results in terms of the MGF instead of the characteristic function. However, we stress that all results of this work can easily accommodate the characteristic function instead.}\\

The above result already illustrates the advantage of using a Fourier representation of risk measures, since the above result applies for any \(p\text{-th}\) moment and any distribution with known MGF. But the result is far more interesting when considering the multivariate case. Indeed, if \((Y_1,\cdots,Y_n)\) is a vector of losses and \(S =\sum_{i=1}^nY_i\) is the sum of losses, only under very special loss distributions will the density of \(S\) be known in closed form, whereas the MGF \(\Phi_{(Y_1,\cdots,Y_n)}(z_1,\cdots,z_n)\) of \((Y_1,\cdots,Y_n)\) may be known, so that the MGF of \(S\) is then known since \(\Phi_S(z)=\Phi_{(Y_1,\cdots,Y_n)}(z,\cdots,z)\), and the aforementioned risk measures can be computed up to a one-dimensional integration.\footnote{If \((Y_1,\cdots,Y_n)\) is a random vector with a known density, only for certain specific distributions the density of the sum \(S=\sum_{1}^n Y_i\) is known in closed form (it is the convolution).} Further to this, in the multivariate case it is relevant to compute the moment \textcolor{blue}{of a product of given losses} conditional on one of the loss, often the sum of the losses but it does not have to be the sum, being above a certain threshold. Proceeding as for \eqref{eq:ConditionalMomentYp}, the problem amounts to calculate
\begin{align}
\mathbb{E}[Z^qY^p | Y>y_*]=\frac{\mathbb{E}[Z^qY^p \mathbf{1}_{\lbrace Y>y_* \rbrace}]}{\mathbb{E}[\mathbf{1}_{\lbrace Y>y_*\rbrace}]}, \label{eq:ConditionalMomentZqYp}
\end{align}
\textcolor{blue}{where $Z$ is one element of $(Y_1,\ldots,Y_n)$ or a linear combination of the elements of $(Y_1,\ldots,Y_n)$} and therefore to determine the numerator of the above equation. Note that in \eqref{eq:ConditionalMomentZqYp}, \(Y\) can be either an individual loss, \textit{i.e.}, \(Y_i\) for \(i=1,\cdots, n\), or a sum of losses \(S =\sum_{i=1}^nY_i\). It is presented in the following proposition.
 
\begin{proposition}\label{prop:JointExpectedShortfallFourierTransform}
Let $(Y,Z)\in \mathbb{R}_+^2$ a 2-dimensional random vector and denote by $\Phi_{(Y,Z)}(v_1,v_2)=\mathbb{E}[e^{v_1Y+v_2Z}]$ for $(v_1,v_2)\in \mathbb{R}^2$ its MGF. Let $y_*\in \mathbb{R}_+$,  $(\alpha_1,\alpha_2)\in \mathbb{R}_+^2$. Assume suitable integrability conditions of the functions  $(u_1,u_2) \in \mathbb{R}^2 \to  \frac{\Phi_{(Y,Z)}(\alpha_1-\mathrm{i} u_1,\alpha_2-\mathrm{i} u_2)}{(\alpha_1- \mathrm{i} u_1)^{j+1}(\alpha_2-\mathrm{i} u_2)^{p+1}} $ for $j=0,\ldots, p$. Then 
\begin{align}
\mathbb{E}[Z^q Y^p\mathbf{1}_{\lbrace Y>y_* \rbrace }]= \int_{0}^{+\infty}\int_{0}^{+\infty}\sum_{j=0}^p \binom{p}{j}\frac{y_*^{p-j}c_jc_q}{\pi^2 } \Re\left(\frac{e^{-(\alpha_1 -\mathrm{i} u_1)y_*} \Phi_{(Y,Z)}(\alpha_1-\mathrm{i} u_1,\alpha_2-\mathrm{i} u_2)}{(\alpha_1-\mathrm{i} u_1)^{j+1}(\alpha_2-\mathrm{i} u_2)^{q+1}}\right)du_1du_2, \label{eq:ExpectationZqYp}
\end{align}
with \(c_j=\Gamma(j+1)\) and \(c_q=\Gamma(q+1)\).
\end{proposition}
\begin{proof}
We have 
\begin{align*}
\mathbb{E}[Z^qY^p \mathbf{1}_{\lbrace Y>y_* \rbrace}]=\sum_{j=0}^p \binom{p}{j} \mathbb{E}[Z^q(Y-y_*)_+^j]y_*^{p-j},
\end{align*}
and if we define $V=Y-y_*$  we obtain $\mathbb{E}[Z^q(Y-y_*)_+^j]=\mathbb{E}[Z^q(V)_+^j]$. Using \citet[3.2(3)]{BatemanEderlyi1954}, we reach
\begin{align*}
\mathbb{E}[Z^q(V)_+^j]=\int_0^{+\infty}\int_{-\infty}^{+\infty}\frac{\Gamma(j+1)\Gamma(q+1)}{(2\pi)^2 }\int_{-\infty}^{+\infty}\int_{-\infty}^{+\infty}\frac{e^{(\alpha_1-\mathrm{i} u_1)v  + (\alpha_2-\mathrm{i}u_2)z}}{(\alpha_1-\mathrm{i} u_1)^{j+1}(\alpha_2-\mathrm{i} u_2)^{q+1}}g(v,z)du_1du_2 dvdz,
\end{align*}
with $g(v,z)$ the density of $(V,Z)$ and $\alpha_1>0$, $\alpha_2>0$. Applying Fubini's theorem, we get
\begin{align*}
\mathbb{E}[Z^q(V)_+^j]=\frac{\Gamma(j+1)\Gamma(q+1)}{(2\pi)^2 }\int_{-\infty}^{+\infty}\int_{-\infty}^{+\infty}\frac{ \Phi_{(V,Z)}(\alpha_1-\mathrm{i} u_1,\alpha_2-\mathrm{i} u_2)}{(\alpha_1-\mathrm{i} u_1)^{j+1}(\alpha_2-\mathrm{i} u_2)^{q+1}}du_1du_2,
\end{align*}
with $\Phi_{(V,Z)}(\alpha_1-\mathrm{i} u_1,\alpha_2-\mathrm{i} u_2) =\mathbb{E}[e^{(\alpha_1-\mathrm{i} u_1)V + (\alpha_2-\mathrm{i} u_2)Z}]$ and since we have
\begin{align*}
\Phi_{(V,Z)}(\alpha_1-\mathrm{i} u_1,\alpha_2-\mathrm{i} u_2)=e^{-(\alpha_1 -\mathrm{i} u_1)y_*}\Phi_{(Y,Z)}(\alpha_1-\mathrm{i} u_1,\alpha_2-\mathrm{i} u_2),
\end{align*}
it leads to
\begin{align*}
\mathbb{E}[Z^q(Y-y_*)_+^j]=\frac{\Gamma(j+1)\Gamma(q+1)}{(2\pi)^2 }\int_{-\infty}^{+\infty}\int_{-\infty}^{+\infty}\frac{ e^{-(\alpha_1 -\mathrm{i} u_1)y_*}\Phi_{(Y,Z)}(\alpha_1-\mathrm{i} u_1,\alpha_2-\mathrm{i} u_2)}{(\alpha_1-\mathrm{i} u_1)^{j+1}(\alpha_2-\mathrm{i} u_2)^{q+1}}du_1du_2,
\end{align*}
that provides the result after considering the real part of the right hand side of the above equation since the left hand side is real. Note that Fubini's theorem requires the integrability of $ (u_1,u_2) \in \mathbb{R}^2 \to \frac{\Phi_{(Y,Z)}(\alpha_1-\mathrm{i} u_1,\alpha_2-\mathrm{i} u_2)}{(\alpha_1-\mathrm{i} u_1)^{j+1}(\alpha_2-\mathrm{i} u_2)^{q+1}}$.
\end{proof}
{\color{black} Note that the previous two results apply to distributions defined not only on \(\mathbb{R}_+^2\), but also on \(\mathbb{R}^2\), provided the integrability conditions are satisfied and the following a minor adjustment to the proof is added. If \(Z\) has a distribution defined on \(\mathbb{R}\), then writing \(z = (z)_+ - (-z)_+\) so that \(z^q = (z)_+^q - (-z)_+^q\), the first term is handled using \citet[3.2(3)]{BatemanEderlyi1954} while the second term is handled using \citet[3.2(4)]{BatemanEderlyi1954}.}\\

{\color{black}
\begin{remark}
Underlying \eqref{eq:MomentYp} and \eqref{eq:ExpectationZqYp} is the Plancherel-Parseval equality which states that the integral of the product of two functions is equal, up to a scaling parameter, to the integral of the product of the Fourier (or Laplace, Mellin) transform of these functions, see \citet[Section 1.14]{DLMF}. This equality is of practical importance since if the integral represents an expectation of a given function of a random variable, then this expectation can be expressed as the product of the function of the random variable and the density of this random variable. The Plancherel-Parseval equality implies that this expectation is given by the integral of the Fourier transform of the density, which is the characteristic function of the random variable, and the Fourier transform of the given function. If the density of the random variable is known, then the expectation can be computed using the density, and there is no need to rely on this Fourier representation. However, it can become quickly tedious. For example, \citet{IgnatievaLandsman2021} compute the TCE for a class of generalized hyper-elliptical distributions, while in \citet{IgnatievaLandsman2025} they extend the computations to the second moment, that is, the tail variance. Higher order (conditional) moments have yet to be determined. If the density of the random variable is unknown (or challenging to compute numerically) but its MGF is known, there are two alternatives. The first consists in computing the density using the inverse
Fourier transform of the characteristic function and then to integrate (numerically) the density. This is done in \citet{NguyenNguyen2017} or \citet{Blier-WongCossetteMarceau2025}. Although it is a possible route, \eqref{eq:MomentYp} and \eqref{eq:ExpectationZqYp} suggest a better alternative solution, which consists in using the Plancherel-Parseval and integrate the product. To this end, one needs to compute the Fourier transform of function involved in the risk measure. For the TCE, this computation is well-known since it is standard in the option pricing literature, and not surprisingly, in \citet{DufresneGarridoMorales2009}, \citet{Bormetti2010} or \citet{KelaniQuittardPinon2014}, this link to this literature is explicitly underlined. As a result, our work extends the literature in two ways. First, we consider higher (conditional) moments as \eqref{eq:MomentYp} clearly shows. Second, we consider the multivariate case in \eqref{eq:ExpectationZqYp}. Note that this result applies to variables with a general distribution, particularly when these variables are the sum of other random variables, as is typical when considering a portfolio. In such situations, a computational approach based on the density can be problematic. Also, multidimensional problem often come with computational difficulties that we partially solve. We elaborate further on this point in the next section.
\end{remark}
}

The above result requires a 2-dimensional integration, and if we consider more variables, but still conditioning on one of them, it will lead to a multi-dimensional Fourier transform (of dimension equal to the number of variables considered), and therefore to a numerical difficulty similar to that encountered when working with the density of a state variable used for the losses. The next section shows how this numerical difficulty can be replaced by a one-dimensional integration and derivation of a function that is related to the derivative of the MGF of the state variable. The difference is important since we have seen in the analytical section that the MGF of the Wishart process is easy to derive.

\subsection{One-dimensional Fourier transform representations for risk measures}\label{sec:OneFourierTransform}

In this section, we reformulate Proposition \ref{prop:JointExpectedShortfallFourierTransform} in a more efficient way from a numerical point of view. Essentially, it requires to be able to compute the derivative of the MGF of the state variable {\color{black}which can typically be done analytically}. 

\begin{proposition}\label{prop:JointExpectedShortfallDerivativeFourierTransform}
Let $(Y,Z)\in \mathbb{R}_+^2$ the 2-dimensional vector of proposition \ref{prop:JointExpectedShortfallFourierTransform} and $\Phi_{(Y,Z)}(v_1,v_2)$ its MGF of $(Y,Z)$. Denote by $\Phi_{(Y,Z)}^{(0,q)}(v_1,v_2)$ for $(v_1,v_2)\in \mathbb{R}^2$ its $q\text{-th}$ order derivative, where $q\in \mathbb{R}$, with respect to $v_2$ (and \(0\text{-th}\) order derivative with respect to $v_1$). Let $y_*\in \mathbb{R}_+$,  $\alpha_1\in \mathbb{R}_+$ and assume suitable integrability conditions of the functions $u_1 \in \mathbb{R} \to \frac{\Phi_{(Y,Z)}^{(0,q)}(\alpha_1-\mathrm{i} u_1,0)}{(\alpha_1-\mathrm{i} u_1)^{j+1}}$ for $j=0,\ldots,p$. Then the following relation holds

\begin{align}
\mathbb{E}[Z^q Y^p\mathbf{1}_{\lbrace Y>y_* \rbrace }]= \int_{0}^{+\infty}\sum_{j=0}^p \binom{p}{j}  \frac{y_*^{p-j}\Gamma(j+1)}{\pi }\Re\left( \frac{e^{-(\alpha_1 -\mathrm{i} u_1)y_* }\Phi_{(Y,Z)}^{(0,q)}(\alpha_1-\mathrm{i} u_1,0)}{(\alpha_1-\mathrm{i} u_1)^{j+1}}\right)du_1. \label{eq:ExpectationZqYp-1dim}
\end{align}

\end{proposition}
\begin{proof}
Proceeding as in the previous proof but not using a Fourier transform for the variable $Z$, we can reach
\begin{align*}
\mathbb{E}[Z^q(Y-y_*)_+^j]=\frac{\Gamma(j+1)}{2\pi }\int_{-\infty}^{+\infty}e^{-(\alpha_1 -\mathrm{i} u_1)y_*} \frac{\mathbb{E}[ Z^qe^{(\alpha_1-\mathrm{i} u_1)Y}]}{(\alpha_1-\mathrm{i} u_1)^{j+1}}du_1,
\end{align*}
with $\alpha_1>0$. Rewrite $\mathbb{E}[ Z^qe^{(\alpha_1-\mathrm{i} u_1)Y}]=\frac{d^q \mathbb{E}\left[ e^{ (\alpha_1-\mathrm{i} u_1)Y + \nu Z}\right] }{d\nu^q}|_{\nu=0} =\Phi_{(Y,Z)}^{(0,q)}(\alpha_1-\mathrm{i} u_1,0)$ and substituting this expression in the equation above we reach the result after taking the real part of the integral since the left hand side is real. Applying Fubini's theorem requires the integrability of $ u_1 \in \mathbb{R} \to \frac{\mathbb{E}[ Z^qe^{(\alpha_1-\mathrm{i} u_1)Y}]}{(\alpha_1-\mathrm{i} u_1)^{j+1}}$.
\end{proof}
{\color{black} Note that the above result applies to a variable defined on \(\mathbb{R}^2\).}\\

The computation of the conditional higher moments \eqref{eq:ConditionalMomentZqYp} can be performed using a one-dimensional integration \eqref{eq:ExpectationZqYp-1dim} instead of a two-dimensional integration \eqref{eq:ExpectationZqYp}. This result can be further extended to more variables. The proposition below provides a method to compute the cross moments of two variables, conditional on a third (dependent) variable. It is presented without proof, as it follows the same structure as the proof of Proposition~\ref{prop:JointExpectedShortfallDerivativeFourierTransform}.

\begin{proposition}\label{prop:TripleExpectedShortfallDerivativeFourierTransform}
	Let $(Y,X,Z)\in \mathbb{R}_+^3$ be a 3-dimensional random vector and denote for $(v_1,v_2,v_3)\in \mathbb{R}^3$ its MGF by $\Phi_{(Y,X,Z)}(v_1,v_2,v_3)$. Denote $\Phi_{(Y,X,Z)}^{(0,r,q)}(v_1,v_2,v_3)$ for $(v_1,v_2,v_3)\in \mathbb{R}^3$ the $r\text{-th}$ order derivative with respect to $v_2$ and the $q$ order derivative with respect to $v_3$ of $\Phi_{(Y,X,Z)}(v_1,v_2,v_3)$ (and $0$ order derivative with respect to $v_1$). Let $y_*\in \mathbb{R}_+$,  $\alpha_1\in \mathbb{R}_+$ and assume suitable integrability conditions of the functions $u_1 \in \mathbb{R} \to \frac{\Phi_{(Y,X,Z)}^{(0,r,q)}(\alpha_1-\mathrm{i} u_1,0,0)}{(\alpha_1-\mathrm{i} u_1)^{j+1}}$ for $j=0,\ldots,p$. Then we have 
	
	\begin{align}
	\mathbb{E}[Y^p  X^r Z^q \mathbf{1}_{\lbrace Y>y_* \rbrace }]= \int_{0}^{+\infty} \sum_{j=0}^p \binom{p}{j}  \frac{y_*^{p-j}\Gamma(j+1)}{\pi }\Re\left( \frac{e^{-(\alpha_1 -\mathrm{i} u_1)y_* }\Phi_{(Y,X,Z)}^{(0,r,q)}(\alpha_1-\mathrm{i} u_1,0,0)}{(\alpha_1-\mathrm{i} u_1)^{j+1}}\right)du_1.\label{eq:ExpectationYXZ}
	\end{align}
	
	\end{proposition}
{\color{black} Note that the above result applies to a variable defined on \(\mathbb{R}^3\).}\\

Note that \eqref{eq:ExpectationYXZ} enables the computation of \(\mathbb{E}[X^rZ^q | Y>y_*]\) from which we can obtain the tail joint central (higher) moment used in \citet[][Eq.~(1.5)]{YangWangYao2025} that extends the tail covariance risk measure of \citet{FurmanLandsman2006}. The purpose of the following remark is mainly to introduce notations that will be needed later in the work.

	\begin{remark}\label{rem:nExpectedShortfallDerivativeFourierTransform}

		Let $(Y_1,...,Y_{n}) \in \mathbb{R}_+^{n}$ be an $n$-dimensional random vector and denote $\Phi_{(Y_1,...,Y_{n})}(v_1,...,v_{n}) = \mathbb{E}[e^{v_1Y_1 + ... + v_{n} Y_{n}}]$ for $(v_1,...,v_{n}) \in \mathbb{R}^{n}$ its MGF. Denote $\Phi_{(Y_1,...,Y_{n})}^{(0,p_2,...,p_{n})}(v_1,...,v_{n})$ the $\left(\sum_{i=2}^{n} p_i\right)\text{-th}$ order derivative of $\Phi_{(Y_1,...,Y_{n})}(v_1,...,v_{n})$, where $(0,p_2,...,p_{n}) \in \mathbb{N}^{n}$ and we differentiate $p_j$ times with respect to $v_j$ for $j=2,\ldots, n$ (note that we do not derive with respect to $v_1$). Let $y_*\in \mathbb{R}_+$,  $\alpha_1\in \mathbb{R}_+$ and assume suitable integrability conditions of the functions $u_1 \in \mathbb{R} \to \frac{\Phi_{(Y_1,...,Y_{\ell})}^{(0,p_2,...,p_{n})}(\alpha_1-\mathrm{i} u_1,0,...,0)}{(\alpha_1-\mathrm{i} u_1)^{j+1}}$ for $j=0,\ldots,p_1$. The following relation holds		

		\begin{align}
			\mathbb{E}\left[\prod_{i=1}^{n} Y_i^{p_i} \mathbf{1}_{\lbrace Y_1>y_* \rbrace }  \right]= \int_{0}^{+\infty}\sum_{j=0}^{p_1} \binom{p_1}{j}  \frac{y_*^{p_1-j}c_j}{\pi }\Re\left( \frac{e^{-(\alpha_1 -\mathrm{i} u_1)y_* }\Phi_{(Y_1,...,Y_{n})}^{(0,p_2,...,p_{n})}(\alpha_1-\mathrm{i} u_1,0,...,0)}{(\alpha_1-\mathrm{i} u_1)^{j+1}}\right)du_1,
			\end{align}
with \(c_j=	\Gamma(j+1)\). {\color{black} Note also that result can be extended without any change to a variable defined on \(\mathbb{R}^n\).}
		\end{remark}

From \eqref{eq:ConditionalMomentYp}, it should be noted that any form of higher order moment or central moment can be constructed. These higher order moments are beneficial for analyzing tail risk properties. For example, the well known tail variance is defined as 
\[ 
\mathbb{E}[(Y - \mathbb{E}[Y|Y > y_*])^2 | Y > y_*] = \mathbb{E}[Y^2|Y > y_*] - \mathbb{E}[Y|Y > y_*]^2, 
\] 
and both terms can be calculated using Proposition~\ref{prop:ExpectedShortfallFourierTransform}. Note that this statement applies whether \(Y\) is a sum of losses or not.\\

The previous propositions apply to any random vector with a known MGF {\color{black} and are particularly relevant when one of the variables is a sum of random variables, as is the case with the sum of losses in a portfolio. Note that only under very specific distributions for the losses will the sum have a tractable density. This makes any methodology based on the density difficult, if not impossible, to implement.} In the next section, we show how the previous expressions can be specified when working with the Wishart process. Interestingly, its process property (\textit{i.e.}, its time dependence) and its positive definite matrix property offer new perspectives that can be effectively addressed using the Fourier transform representation of risk measures.

\subsection{Wishart conditional tail risk measures}\label{sec:WishartRiskMeasures}
In the previous section, the results were presented in a very general form. Essentially, they require the MGF, the derivative of the MGF of the state variable, as well as some integrability conditions for these functions. The purpose of this section is to show that the Wishart process satisfies these requirements.\\

Let $(x_{11,t},\ldots, x_{nn,t})$ be an $n\times 1$ vector of losses, and assume that this vector is the diagonal of a Wishart process $x_t$ of size $n$. Note that $x_{ii,t}=\textup{tr}[\theta_1 x_t]$ with $\theta_1 = e_{ii}$. A standard risk measure associated with $x_{ii,t}$ is the \(q\text{-th}\) moment of $x_{ii,t}$ conditional on \( x_{ii,t}>x_* \) for a given $x_*\in \mathbb{R}_+$. It is given by the following proposition.
\begin{proposition}\label{prop:ConditionalMomentsOfX}
Assume that $\omega=\beta \sigma^2$ and let $\theta_1= e_{ii}$ for $i\in \lbrace 1,\ldots,n \rbrace $ so that  $x_{ii,t}=\textup{tr}[\theta_1 x_t]$, then the $q\text{-th}$ moment of $\textup{tr}[\theta_1 x_t]$ conditional on $\textup{tr}[\theta_1 x_t]>x_*$, with $x_*\in \mathbb{R}_+$,  is given by
\begin{align}
\mathbb{E}\left[ \textup{tr}[\theta_1 x_t]^q |  \textup{tr}[\theta_1 x_t]>x_*\right]=\frac{\int_0^{+\infty}\sum_{j=0}^q\binom{q}{j} x_*^{q-j}c_j\Re\left( \frac{e^{-(\alpha_1-\mathrm{i}u_1)x_*}\Phi_{\textup{tr}[\theta_1x_t]}(\alpha_1-\mathrm{i}u_1)}{(\alpha_1-\mathrm{i}u_1)^{j+1}} \right)du_1}{\int_0^{+\infty}\Re\left( \frac{e^{-(\alpha_1-\mathrm{i}u_1)x_*}\Phi_{\textup{tr}[\theta_1x_t]}(\alpha_1-\mathrm{i}u_1)}{(\alpha_1-\mathrm{i}u_1)} \right)du_1},\label{eq:ConditionalMomentXq}
\end{align}
with $c_j=\Gamma(j+1)$, $\Phi_{\textup{tr}[\theta_1x_t]}(\alpha_1-\mathrm{i}u_1)=\Phi(t, (\alpha_1-\mathrm{i}u_1)\theta_1,x_0)$, $\Phi(.)$ given by \eqref{eq:mgf} or \eqref{eq:solPhi} (with $a(t,\theta_0)$ as in \eqref{eq:Asimple} and $b(t)$ as in \eqref{eq:ExpB}) and $\alpha_1\in \mathbb{R}_+$.
\end{proposition}

\begin{proof}
The result is an immediate application of \eqref{eq:ConditionalMomentYp} and \eqref{eq:MomentYp}. It remains to check the integrability conditions of Proposition~\ref{prop:ExpectedShortfallFourierTransform}. For $j=1,\ldots, q$, we have $|e^{-(\alpha_1-\mathrm{i}u_1)x_*}\Phi_{\textup{tr}[\theta_1x_t]}(\alpha_1-\mathrm{i}u_1)/(\alpha_1-\mathrm{i}u_1)^{j+1}|\leq e^{-\alpha_1x_*}|\mathbb{E}[e^{(\alpha_1-\mathrm{i}u_1)\textup{tr}[\theta_1x_t]}]|/|(\alpha_1-\mathrm{i}u_1)^{j+1}|\leq c \mathbb{E}[e^{\alpha_1\textup{tr}[\theta_1x_t]}]/|(\alpha_1-\mathrm{i}u_1)^{j+1}|$ for a certain constant $c>0$ after using Jensen's inequality. If we can prove that $\mathbb{E}[e^{\alpha_1\textup{tr}[\theta_1x_t]}]<\infty$ then the integrals of the numerator \eqref{eq:ConditionalMomentXq} are well-defined (for $j=1,\ldots,q$) since the integrands are bounded by an integrable function. Replacing $\theta_1$ with $\alpha_1\theta_1$ in \eqref{eq:solPhi}, we get for \eqref{eq:Asimple} that $|a(t,\alpha_1\theta_1)|<\infty$ for $\alpha_1$ small enough since $(I_n-2\alpha_1\theta_1\varsigma_t)^{-1}$ is well-defined because $\varsigma_t$ is bounded. The same argument can be used to prove that $|\det(I_n-2\alpha_1\theta_1\varsigma_t)|>0$ in \eqref{eq:ExpB} so that $|b(t,\alpha_1\theta_1)|<\infty$. Combining these two arguments we obtain the announced integrability condition. For the $j=0$, and in particular the denominator of \eqref{eq:ConditionalMomentXq}, it requires a sharper upper bound. Replace $\theta_1$ in \eqref{eq:Asimple} with $(\alpha_1 -\mathrm{i} u_1)\theta_1$ then the function $u_1 \to |e^{tm^\top}(I_n-2(\alpha_1 -\mathrm{i}u_1)\theta_1\varsigma_t)^{-1}(\alpha_1 -\mathrm{i}u_1)\theta_1e^{tm})|$ is bounded by a constant when $u_1\to +\infty$ as is $u_1 \to |a(t,(\alpha_1-\mathrm{i}u_1)\theta_1)|$. Further to this, $u_1\to 1/|\det(I_n - 2(\alpha_1 -\mathrm{i}u_1)\theta_1\varsigma_t)|^{\beta/2}$ behaves like $1/u_1^{n\beta/2}$ when $u_1\to +\infty$ so that the integrand in the denominator of \eqref{eq:ConditionalMomentXq}, or the numerator when $j=0$, behaves like $1/u_1^{n\beta/2 + 1}$. But since $\beta\geq n+1$ we conclude that these integrals are well-defined.  

\end{proof}

Note that when \(q=1\)  in \eqref{eq:ConditionalMomentXq}, then the above conditional expectation corresponds to $\TCE_{x_{ii,t}}(x_*)=\mathbb{E}\left[ x_{ii,t} |  x_{ii,t}>x_* \right]$. If we denote by $s_t = \sum_{i=1}^n x_{ii,t}$ the sum of those losses on a portfolio, it can also be written as $\textup{tr}[\theta_1 x_t]$ with $\theta_1 = I_n$. The previous proposition thus gives 
\begin{align}
\mathbb{E}\left[ s_t^q |  s_t>s_* \right], \label{eq:ConditionalMomentSq}
\end{align}
with $s_*>0$, it is the $q\text{-th}$ moment of the losses on a portfolio conditional on the sum of these being above a certain threshold. For $q=1$, it is \(\TCE_{s_t}(s_*)\), the TCE of the portfolio.\\

When considering several losses, the quantity of interest is the expected loss of the variable \(x_{ii,t}\) conditional on the sum of losses being above a certain threshold, it is given by $\mathbb{E}[x_{ii,t}| s_t>s_*]$. An extension of this measure is the higher moments of \(x_{ii,t}\) conditional on the sum of losses being above a certain threshold, \textit{i.e.}, $\mathbb{E}[x_{ii,t}^q| s_t>s_*]$. But we can further extend from higher moments of \(x_{ii,t}\) to higher cross-moments of \((x_{ii,t},s_t)\) conditional on the sum of losses being above a certain threshold, \textit{i.e.}, \(\mathbb{E}[x_{ii,t}^qs_t^p | s_t>s_*]\). Thanks to the previous analytical results, these quantities are easy to obtain in our framework. In particular, the following proposition derives from Proposition \ref{prop:JointExpectedShortfallFourierTransform}.
\begin{proposition}\label{prop:MomentsOfXConditionalOnS}
Assume that $\omega=\beta \sigma^2$ and let $\theta_0=\sum_{i=1}^{n} e_{ii}$  and $\theta_1= e_{ii}$ for $i\in \lbrace 1,\ldots,n \rbrace $ so that $s_t=\textup{tr}[\theta_0 x_t]$ and   $x_{ii,t}=\textup{tr}[\theta_1 x_t]$. Then the $(p\text{-th},q\text{-th})$ moment of $(\textup{tr}[\theta_0 x_t],\textup{tr}[\theta_1 x_t])$ conditional on the event \(s_t>s_*\), with \(s_*\in \mathbb{R}_+\),  is 
\begin{align}
I &= \mathbb{E}\left[ \textup{tr}[\theta_1 x_t]^q\textup{tr}[\theta_0 x_t]^p | \textup{tr}[\theta_0 x_t]>s_* \right] \nonumber\\
  &=\frac{\int_{0}^{+\infty}\int_{0}^{+\infty} \sum_{j=0}^p \frac{1}{\pi}\binom{p}{j} s_*^{p-j}c_{j}c_{q} \Re\left(\frac{e^{-(\alpha_0 -\mathrm{i} u_0)s_*} \Phi_{(\textup{tr}[\theta_0 x_t],\textup{tr}[\theta_1 x_t])}(\alpha_0-\mathrm{i} u_0,\alpha_1-\mathrm{i} u_1)}{(\alpha_0-\mathrm{i} u_0)^{j+1}(\alpha_1-\mathrm{i} u_1)^{q+1}}\right)du_0du_1}{\int_0^{+\infty}\Re\left( \frac{e^{-(\alpha_0 -\mathrm{i} u_0)s_*}\Phi_{\textup{tr}[\theta_0x_t]}(\alpha_0-\mathrm{i}u_0)}{(\alpha_0-\mathrm{i}u_0)} \right)du_0},\label{eq:MomentsOfXConditionalOnS}
\end{align}
\normalsize
with $c_{j}=\Gamma(j+1)$, $c_q=\Gamma(q+1)$, $\Phi_{(\textup{tr}[\theta_0x_t],\textup{tr}[\theta_1x_t])}(\alpha_0-\mathrm{i}u_0,\alpha_1-\mathrm{i}u_1)=\Phi(t, (\alpha_0-\mathrm{i}u_0)\theta_0 + (\alpha_1-\mathrm{i}u_1)\theta_1,x_0)$, $\Phi(.)$ given by \eqref{eq:mgf}, $(\alpha_0,\alpha_1)\in \mathbb{R}_+^2$.
\end{proposition}
\begin{proof}
The integrability condition when $j=1,\ldots,p$ is:
\begin{align*}
\left| \frac{ e^{-(\alpha_0 -\mathrm{i} u_0)s_* }\Phi_{(\textup{tr}[\theta_0 x_t],\textup{tr}[\theta_1 x_t])}(\alpha_0-\mathrm{i} u_0,\alpha_1-\mathrm{i} u_1)}{(\alpha_0-\mathrm{i} u_0)^{j+1}(\alpha_1-\mathrm{i} u_1)^{q+1}}\right| & \leq \frac{e^{-\alpha_0s_* }|\mathbb{E}[e^{(\alpha_0-\mathrm{i} u_0)\textup{tr}[\theta_0x_t] + (\alpha_1-\mathrm{i} u_1)\textup{tr}[\theta_1x_t] }]|}{| (\alpha_0-\mathrm{i} u_0)^{j+1}(\alpha_1-\mathrm{i} u_1)^{q+1}|}\\
&\leq e^{-\alpha_0 s_* } \frac{|\mathbb{E}[e^{\alpha_0\textup{tr}[\theta_0x_t] + \Re(\alpha_1)\textup{tr}[\theta_1x_t] }]|}{| (\alpha_0-\mathrm{i} u_0)^{j+1}(\alpha_1-\mathrm{i} u_1)^{q+1}|},
\end{align*}
after using Jensen's inequality. Taking $\alpha_0>0$ and $\alpha_1>0$ but sufficiently small, using similar arguments as in the previous proposition we reach that $\mathbb{E}[e^{\textup{tr}[(\alpha_0\theta_0 + \alpha_1\theta_1)x_t] }]<\infty$ and, therefore, that the integrand in the numerator of \eqref{eq:MomentsOfXConditionalOnS} is bounded by an integrable function. When $j=0$, we proceed as in the previous proposition to establish the integrability.
\end{proof}

As previously mentioned, the above result is of interest but it involves a multi-dimensional integration. However, Proposition~\ref{prop:JointExpectedShortfallDerivativeFourierTransform} shows that this can be reduced to a one-dimensional integration if the MGF of the underlying state variable can be derived. As Corollaries~\ref{coro:MGFWishartQDerivative} and \ref{coro:MGF2WishartQDerivativeTime2} clearly demonstrate, the Wishart process possesses this property. The proof of the following proposition is a straightforward consequence of Proposition~\ref{prop:JointExpectedShortfallDerivativeFourierTransform} and is therefore omitted.

\begin{proposition}\label{prop:MomentsOfXConditionalOnS1Dim}
Let $\theta_0=\sum_{i=1}^{n} e_{ii}$  and $\theta_1= e_{ii}$ for $i\in \lbrace 1,\ldots,n \rbrace $ so that $s_t=\textup{tr}[\theta_0 x_t]$ and   $x_{ii,t}=\textup{tr}[\theta_1 x_t]$. Then the $(p\text{-th},q\text{-th})$ moment of $(\textup{tr}[\theta_0 x_t],\textup{tr}[\theta_1 x_t])$ conditional on the event $ s_t>s_*$, with $s_*\in \mathbb{R}_+$, can be written as

\begin{align}
\mathbb{E}\left[ \textup{tr}[\theta_0 x_t]^p\textup{tr}[\theta_1 x_t]^q | \textup{tr}[\theta_0 x_t]>s_* \right]=\frac{\int_{0}^{+\infty} \sum_{j=0}^p \binom{p}{j}  s_*^{p-j}c_j\Re\left( \frac{e^{-(\alpha_0 -\mathrm{i} u_0)s_* }\Phi_{(\textup{tr}[\theta_0x_t],\textup{tr}[\theta_1x_t])}^{(0,q)}(\alpha_0-\mathrm{i} u_0,0)}{(\alpha_0-\mathrm{i} u_0)^{j+1}}\right)du_0}{\int_0^{+\infty}\Re\left( \frac{e^{-(\alpha_0 -\mathrm{i} u_0)s_*}\Phi_{\textup{tr}[\theta_0x_t]}(\alpha_0-\mathrm{i}u_0)}{(\alpha_0-\mathrm{i}u_0)} \right)du_0},\label{eq:ConditionalMomentXq}
\end{align}
\normalsize
with \(\alpha_0>0\), \(c_j=\Gamma(j+1)\), $\Phi^{(0,q)}_{(\textup{tr}[\theta_0x_t],\textup{tr}[\theta_1x_t])}(\alpha_0-\mathrm{i}u_0,0)=\bar{\Phi}_{\nu^q}(t, (\alpha_0-\mathrm{i}u_0)\theta_0, \theta_1,x_0)$ and $\bar{\Phi}_{\nu^q}(.,.,.,.)$ given by \eqref{eq:MGFWishartQDerivative}.
\end{proposition}

Note that no order is specified on $\theta_1$ and $\theta_2$, and so we can interchange these matrices at will. This means we can easily compute the cross-moments by conditioning on the individual losses $x_{ii,t}$ or on the portfolio itself $\sum_{i=1}^n x_{ii,t}$.\\

Still in a portfolio perspective, in our framework we not only specify the distribution of losses but also their dependencies, which can be used to perform conditioning. Let $\theta_0= (e_{12} + e_{21})/2$, then $x_{12,t}=\textup{tr}[\theta_0 x_t]$. Similarly, let $\theta_1= e_{11}$ and $\theta_2= e_{22}$, so that $x_{11,t}=\textup{tr}[\theta_1 x_t]$ and $x_{22,t}=\textup{tr}[\theta_2 x_t]$. We can then calculate 
\begin{align} 
\mathbb{E}[x_{11,t}^p x_{22,t}^q| x_{12,t}>x_*], 
\end{align} which is the expected cross-moment of the first loss and the second loss conditional on the covariance between these losses being above a certain threshold. This example exploits the positive definite matrix nature of the Wishart process. The proposition below is an application of Proposition~\ref{prop:TripleExpectedShortfallDerivativeFourierTransform}, which allows us to compute such a conditional cross-moment of the Wishart process.

\begin{proposition}\label{prop:LMomentsOfXConditionalOnS1Dim}
Consider $\theta_0= (e_{12} + e_{21})/2$, $\theta_1= e_{11}$, and $\theta_2= e_{22}$ so that $x_{12,t}=\textup{tr}[\theta_0 x_t]$, $x_{11,t}=\textup{tr}[\theta_1 x_t]$ and $x_{22,t}=\textup{tr}[\theta_2 x_t]$. For a given pair $(p,q)\in \mathbb{N}^2$ and $x_* \in \mathbb{R}_+$, we have
	\begin{align}
	\mathbb{E}\left[ \textup{tr}[\theta_1 x_t]^{p}\textup{tr}[\theta_{2} x_t]^{q} | \textup{tr}[\theta_0 x_t]>x_* \right]=\frac{ \int_{0}^{+\infty}\Re\left( \frac{e^{-(\alpha_0 -\mathrm{i} u_0)x_* }\Phi_{(\textup{tr}[\theta_0x_t],\textup{tr}[\theta_1x_t],\textup{tr}[\theta_{2}x_t])}^{(0,p,q)}(\alpha_0-\mathrm{i} u_0,0,0)}{(\alpha_0-\mathrm{i} u_0)}\right)du_0}{\int_0^{+\infty}\Re\left( \frac{e^{-(\alpha_0 -\mathrm{i} u_0)x_*}\Phi_{\textup{tr}[\theta_0x_t]}(\alpha_0-\mathrm{i}u_0)}{(\alpha_0-\mathrm{i}u_0)} \right)du_0},\label{eq:nConditionalMomentXq}
	\end{align}
	with  $\Phi^{(0,p,q)}_{(\textup{tr}[\theta_0x_t],\textup{tr}[\theta_1x_t],\textup{tr}[\theta_{2}x_t])}(\alpha_0-\mathrm{i}u_0,0,0) = \left.\frac{\partial^{p+q}}{\partial \nu_1^p \partial \nu_{2}^q} \mathbb{E}\left[ e^{(\alpha_1 - \mathrm{i} u_1)\textup{tr}[\theta_0 x_t] + \nu_1 \textup{tr}[\theta_1 x_t] +\nu_{2}\textup{tr}[\theta_{2}x_t]   }\right]\right|_{\nu_{1} = \nu_{2} = 0}$ whose calculation can be done following those of Proposition~\ref{prop:MGFWishartDerivative} and Corollary~\ref{coro:MGFWishartQDerivative}.
	\end{proposition}
{\color{black} Note that \(x_{12,t}\) can be either positive or negative. As per the remark that follows Proposition \ref{prop:JointExpectedShortfallFourierTransform}, one has to use either \citet[3.2(3)]{BatemanEderlyi1954}, leading to \(\alpha_0>0\), or \citet[3.2(4)]{BatemanEderlyi1954} leading to \(\alpha_0<0\).}\\

{\color{black}
\begin{remark}
As Equation \eqref{eq:ConditionalMomentSq} shows, Proposition \ref{prop:ConditionalMomentsOfX} applies to a portfolio, while Propositions \ref{prop:MomentsOfXConditionalOnS1Dim} and \ref{prop:LMomentsOfXConditionalOnS1Dim} demonstrate how to compute higher (conditional) cross-moments. The Wishart distribution is an example of a distribution for which computing these risk measures using the density leads to significant numerical difficulties, and in particular when the conditioning variable in the sum of dependent losses. Note also that the framework, both the methodology and the chosen distribution on the space of symmetric positive definite matrices, allows for the determination of these risk measures even when one (or several) of the variables controlling the dependency of the distribution is involved. This dependency can be fairly general; in particular, it can take any sign, which contrasts with existing results (see \citet[Definition 1, Condition 5]{FurmanLandsman2005} or \citet[Examples 6.1 and 6.2]{Blier-WongCossetteMarceau2025}), where there is a sign constraint because the state variable is a vector of independent random variables, and the dependence derives from one of the variable being a common factor to the losses.\\

Let us mention the two recent interesting works \citet{AriasSernaCaroLoperaLoubes2021} and \citet{AriasSernaCaroLoperaLoubes2025}, which compute the VaR, not a coherent risk measure, using the density of the state variable following a Wishart distribution. Compared to these works, the above formula give conditional higher moments, and by formulating the problem in terms of the MGF, the methodology remains numerically simple to implement since the density of the Wishart process, besides depending of many variables, involves a special function of matrix argument that is challenging to evaluate. In contrast, the MGF of the Wishart distribution only involves the ratio of matrices and the determinant of matrices, both functions are simple to derive. This aspect is crucial to obtain the above formulas.
\end{remark}
}

The above propositions, while daunting at first glance, are easily computable in most programming languages.

\subsubsection{The two dates case}

We can take advantage of the time-dependent nature of the Wishart process by considering losses at two different dates. In this particular case, one is interested in the intertemporal or time-lagged properties of the risk measures, which are most often analyzed using static distributions. {\color{black} This perspective relates to the time-series literature on risk measures, see for example, \citet{ChavezDemoulinGuillou2018} or \citet{GoegebeurGuillou2024}. Interestingly, it can also be related the literature on (multivariate) count processes (such as the Cox process), where one models two (or more) events (in our case, losses), and their dependence (see \citet{LuZhangZhu2024}).} The temporal nature of this problem often involves leveraging the autocorrelation of a single loss or the correlation between two types of losses occurring at different times. Within the context of the present work, this translates as follows.\\

Given two dates $t_0, t_1 \in \mathbb{R}_+$ such that $0<t_0 <t_1$, we are interested in the quantities $\mathbb{E}\left[ x_{ii,t_1}^q |  x_{ii,t_0}>x_* \right]$, $\mathbb{E}\left[ s_{t_1}^p | s_{t_0}>s_* \right]$, $\mathbb{E}\left[ s_{t_1}^p x_{ii, t_1}^q | s_{t_0}>s_* \right]$ with $x_*>0$, $s_*>0$ and $p, q \in \mathbb{N}$. For example, the first (conditional) expectation represents the expected value of the $q\text{-th}$ moment of the $i^{\text{th}}$ loss at time $t_1$ conditional on the event that at the previous time $t_0$  this loss was greater than $x_*$. The other expectations can be interpreted in a similar way. Thanks to Proposition~\ref{prop:MGFWishart2Dates}, Proposition~\ref{prop:MGF2WishartDerivative2}, Corollary~\ref{coro:MGF2WishartQDerivativeTime2} and Proposition \ref{prop:JointExpectedShortfallDerivativeFourierTransform}, they can be determined explicitly. Let us formulate the result in very general terms, with all these conditional expectations being special cases.\\

\begin{proposition}\label{prop:ConditionalMomentsOfXTwoDates2}
Let $t_0, t_1 \in \mathbb{R}_+$  such that $0<t_0<t_1$ and let $\theta_0,\theta_1\in \mathbb{S}_n$. Then the $(p\text{-th},q\text{-th})$ moment of $(\textup{tr}[\theta_0 x_{t_1}],\textup{tr}[\theta_1 x_{t_1}])$ conditional on $\textup{tr}[\theta_0 x_{t_0}]>x_*$, with $x_*\in \mathbb{R}_+$,  is given by
\begin{align}
\mathbb{E}\left[ \textup{tr}[\theta_0 x_{t_1}]^p \textup{tr}[\theta_1 x_{t_1}]^q | \textup{tr}[\theta_0 x_{t_0}]>x_* \right]=\frac{ \int_{0}^{+\infty}\Re\left( \frac{e^{-(\alpha_0 -\mathrm{i} u_0)x_* }\Phi_{(\textup{tr}[\theta_0x_{t_0}],\textup{tr}[\theta_0x_{t_1}],\textup{tr}[\theta_1x_{t_1}])}^{(0,p,q)}(\alpha_0-\mathrm{i} u_0,0,0)}{(\alpha_0-\mathrm{i} u_0)}\right)du_0}{\int_0^{+\infty}\Re\left( \frac{e^{-(\alpha_0 -\mathrm{i} u_0)x_*}\Phi_{\textup{tr}[\theta_0x_{t_0}]}(\alpha_0-\mathrm{i}u_0)}{(\alpha_0-\mathrm{i}u_0)} \right)du_0},\label{eq:ConditionalMomentXq2}
\end{align}
with $\Phi_{(\textup{tr}[\theta_0x_{t_0}],\textup{tr}[\theta_0x_{t_1}],\textup{tr}[\theta_1x_{t_1}])}^{(0,p,q)}(\alpha_0-\mathrm{i} u_0,0,0) =\tilde{\Phi}_{\nu_0^p\nu_1^q}(t_0,t_1,(\alpha_0-\mathrm{i} u_0)\theta_0,\theta_0,\theta_1 ,x_0) $ and  $\Phi_{\textup{tr}[\theta_0x_{t_0}]}(\alpha_0-\mathrm{i}u_0)=\Phi(t_0, (\alpha_0-\mathrm{i}u_0)\theta_0,x_0)$  given by ~\eqref{eq:MGF2WishartQDerivativeTime2} and ~\eqref{eq:mgf}, respectively, while $\alpha_0\in \mathbb{R}_+$.
\end{proposition}
\begin{proof}
Define \(X=\textup{tr}[\theta_0x_{t_0}]\), \(Z_0=\textup{tr}[\theta_0x_{t_1}]\), and \(Z_1=\textup{tr}[\theta_1x_{t_1}]\). Then, proceeding as in Proposition \ref{prop:TripleExpectedShortfallDerivativeFourierTransform}, we reach
\begin{align*}
\mathbb{E}\left[ \textup{tr}[\theta_0 x_{t_1}]^p \textup{tr}[\theta_1 x_{t_1}]^q \mathbf{1}_{\lbrace \textup{tr}[\theta_0 x_{t_0}]>x_* \rbrace}\right]&=\mathbb{E}\left[Z_0^pZ_1^q \mathbf{1}_{\lbrace X>x_* \rbrace}\right]\\
&=\frac{1}{\pi}\int_0^{+\infty}\Re\left( \frac{e^{-(\alpha_0-\mathrm{i}u_0)x_*}\mathbb{E}[Z_0^pZ_1^qe^{(\alpha_0-\mathrm{i}u_0)X} ] }{(\alpha_0-\mathrm{i}u_0)} \right)du_0\\
&=\frac{1}{\pi}\int_0^{+\infty}\Re\left( \frac{e^{-(\alpha_0-\mathrm{i}u_0)x_*}\partial_{\nu_0^p\nu_1^q}^{p+q} \mathbb{E}[e^{(\alpha_0-\mathrm{i}u_0)X + \nu_0 Z_0+\nu_1Z_1} ] |_{\nu_0=0,\nu_1=0}}{(\alpha_0-\mathrm{i}u_0)} \right)du_0\,,
\end{align*}
with \(\alpha_0>0\). We conclude after defining, for \((v,v_0,v_1)\in \mathbb{R}^3\), the MGF \(\Phi_{(X,Z_0,Z_1)}(v,v_0,v_1) =\mathbb{E}[e^{v X + v_0 Z_0+v_1Z_1} ]= \tilde{\Phi}(t_0,t_1,v\theta_0,v_0,v_1,\theta_0,\theta_1,x_0)\), with this function given by \eqref{eq:barPhiTime2} and its \((0,p,q)\) derivative with respect to \((v,v_0,v_1)\), which we call \(\Phi_{(X,Z_0,Z_1)}^{(0,p,q)}(v,v_0,v_1)\). This is equal, for \(v=\alpha_0 -\mathrm{i}u_0\), \(v_0=0\), and \(v_1=0\), to \(\tilde{\Phi}_{\nu_0^p\nu_1^q}(t_0,t_1,(\alpha_0-\mathrm{i} u_0)\theta_0,\theta_0,\theta_1 ,x_0)\) given by \eqref{eq:MGF2WishartQDerivativeTime2}. The denominator in \eqref{eq:ConditionalMomentXq2} is obtained for \(p=0\) and \(q=0\).
\end{proof}

\subsection{A note on numerical precision}

In this section, we contextualize our results within the broader literature. We first show how our formulae align with common notational conventions, and then demonstrate how they can be applied to well-known distributions whose MGFs are already established in the literature.\\

A commonly encountered quantity is the TCE of a random variable \( Y \) at a threshold \( y_* \) that is defined as:
\[
\TCE_Y(y_*) = \mathbb{E}[Y \mid Y > y_*].
\]

Using \eqref{eq:ConditionalMomentYp} and Proposition~\ref{prop:ExpectedShortfallFourierTransform}, we can compute \(\TCE_Y(y_*)\) exactly from the MGF \(\Phi_Y(\cdot)\). This approach involves no approximation other than standard numerical integration. Notably, this method also extends to the sum of dependent random variables, \( S = \sum_{i=1}^n Y_i \), provided the MGF $\Phi_{(Y_1,...,Y_n)}(\cdot,...,\cdot)$ is known since \(\Phi_S(z)=\Phi_{(Y_1,...,Y_n)}(z,...,z)\) for \(z\in \mathbb{R}\)  is known. Thus, the approach remains entirely general.\\

One might argue that numerical integration is an approximation. However, it is important to emphasize that most `closed-form' results in the literature already involve integrals that must be evaluated numerically. For example, consider the generalized hyperbolic (GH) distribution of \cite{IgnatievaLandsman2019}, where a random vector is defined to be a GH distribution,  \((Y_1,...,Y_n) \sim GH_{n}(\lambda, \chi, \psi, \boldsymbol{\mu}, \Sigma, \boldsymbol{\gamma})\), if its MGF is given for \(\mathbf{v}\in \mathbb{R}^n\)  by
\begin{align}
    \Phi_{(Y_1,...,Y_n)}(\mathbf{v}) = \exp(\mathbf{v}^{\top} \boldsymbol{\mu}) \left(\frac{\psi}{\psi - 2 \mathbf{v}^\top \boldsymbol{\gamma} - \mathbf{v}^\top \Sigma \mathbf{v}}\right)^{\lambda / 2}  \frac{K_\lambda\left(\sqrt{\chi \left(\psi - 2 \mathbf{v}^\top \boldsymbol{\gamma} - \mathbf{v}^\top \Sigma \mathbf{v}\right)}\right)}{K_\lambda\left(\sqrt{\chi \psi}\right)}, \label{eq:GH_mgf}
\end{align}
where \( K_{\lambda}(\cdot) \) is the modified Bessel function of the third kind,  see \citet[Eq.~(2.6)]{IgnatievaLandsman2019}. Our integral-based approach can directly use this MGF and the known derivatives of \(K_{\lambda}(\cdot)\) (see \citet[][Eq.~(10.29.2)]{DLMF}) to produce all desired tail measures without further approximation. Alternatively, the key result from \cite{IgnatievaLandsman2019} computes the TCE of a (univariate) GH distribution through their Theorem 3.1, which is
\begin{align}
    \mathrm{TCE}_Y(y_*) = \mu + \frac{\gamma}{1-q} 
k_\lambda \bar{F}_{\mathrm{GH}_1}(y_*; \lambda + 1, \chi, \psi, \mu, \sigma^2, \gamma)
+ \frac{\sigma^2}{1-q} 
k_\lambda f_{\mathrm{GH}_1}(y_*; \lambda + 1, \chi, \psi, \mu, \sigma^2, \gamma), \label{eq:TCE_IL_1}
\end{align}
where $y_*$ is the $q$ level quantile and the density $y \to f_{\mathrm{GH}_1}(y; \lambda, \chi, \psi, \mu, \sigma^2, \gamma)$ takes the form of
\begin{align}
    f_{\mathrm{GH}_1}(y; \lambda&, \chi, \psi, \mu, \sigma^2, \gamma) =\\ 
&\frac{\psi^{\lambda} (\psi + \gamma^2 \sigma^{-2})^{\frac{1}{2} - \lambda} \left(\sqrt{\chi \psi}\right)^{-\lambda}}
{\sqrt{2\pi} \sigma K_{\lambda}(\sqrt{\chi \psi})}
\frac{K_{\lambda-\frac{1}{2}} \left( \sqrt{(\chi + Q)(\psi + \gamma^2 \sigma^{-2})} \right)}
{\left( \sqrt{(\chi + Q)(\psi + \gamma^2 \sigma^{-2})} \right)^{\frac{1}{2} - \lambda}}
e^{(y - \mu) \gamma \sigma^{-2}}, \label{eq:GH_density}
\end{align}
where \( \textup{det}(\Sigma)^{\frac{1}{2}} = \sigma \), and \( Q = (y - \mu)^2 \sigma^{-2} \), see \citet[Eq. (2.7)]{IgnatievaLandsman2019}. The survival function $\bar{F}_{\mathrm{GH}_1}(y_q; \lambda + 1, \chi, \psi, \mu, \sigma^2, \gamma)$ of the GH distribution is given by 1 minus the integral of the density function $f_{GH_1}(y_q; \lambda + 1, \chi, \psi, \mu, \sigma^2, \gamma)$ and $k_{\lambda}$ is a constant satisfying 
\begin{align}
k_\lambda = \sqrt{\frac{\chi}{\psi}} \frac{K_{\lambda+1}(\sqrt{\chi \psi})}{K_\lambda(\sqrt{\chi \psi})}. \label{eq:k_lambda}
\end{align}
As we can see, Theorem 3.1 of \cite{IgnatievaLandsman2019} requires an integration due to the survival function $\bar{F}_{GH_1}(y)$, where our results also only require a one-dimensional integration. However, our approach also has the flexibility to calculate conditional higher moments as we shall see below.\\

When considering a sum of losses, the TCE of the sum can be computed thanks to a formula of the form  \eqref{eq:TCE_IL_1} since if \((Y_1,...,Y_n) \sim GH_{n}(\lambda, \chi, \psi, \boldsymbol{\mu}, \Sigma, \boldsymbol{\gamma})\) then  \( S = \sum_{i=1}^n Y_i \) follows the distribution \(GH_{1}(\lambda, \chi, \psi, \textbf{1}^\top \boldsymbol{\mu}, \textbf{1}^\top \Sigma \textbf{1}, \textbf{1}^\top \boldsymbol{\gamma})\) with \(\textbf{1}=(1,\ldots,1)^\top\) an \(n\)-dimensional vector of ones. It is this property of the GH distribution that ensures a simple expression for the density of the sum that keeps the computation of the TCE a one-dimensional integration problem even in the multivariate case. Note that for many distributions, and in particular for the Wishart distribution, the density of the sum does not admit a simple expression, making the computation of the TCE of the sum challenging.\\

Furthermore, \citet[Theorem 3.2]{IgnatievaLandsman2019} show that the TCE of the sum can be decomposed as a sum of expectations of losses conditional on the sum being above a certain threshold. These conditional expectations are of the form \(\mathbb{E}[Y_i | S > s_*]\) and thanks to the property of the GH distribution they can admit an expression of the form \eqref{eq:TCE_IL_1} (with adjusted parameters). From a numerical point of view, it amounts to perform several one-dimensional integrations. With our approach, since the MGF of \((Y_1,...,Y_n)\) is known, the results can also be obtained at the same numerical cost.\myfootnote{{\color{black} A worked-out implementation of the \(\TCE\) for the GH distribution, computed using the density as in \citet{IgnatievaLandsman2019} and our approach, is provided in the supplementary appendix.}}\\

Thus, the numerical cost, the precision and accuracy of our results are comparable to, and in many cases equivalent to, methods currently in use. However, our approach applies to distributions for which the density of the sum is not simple to compute, and the Wishart distribution is one of them.\\

As previously noted, our approach not only applies to a wide range of distributions but also enables the computation of conditional higher moments, as demonstrated in the following section.

\subsection{Conditional higher moments tail risk measures}

Beyond the TCE, our methodology readily extends to other tail-based risk measures, and when considering the risk measures of a sum of losses, it naturally introduces conditional higher moments of losses.

\paragraph{The Tail variance (TV):} The tail variance, defined as the conditional second central moment given \( Y > y_* \), can be written as:
    \begin{align}
    \TV_Y(y_*) = \mathbb{E}[(Y - \mathbb{E}[Y | Y > y_*])^2 \mid Y > y_*] 
    = \mathbb{E}[Y^2 \mid Y > y_*] - (\mathbb{E}[Y \mid Y > y_*])^2. \label{eq:TailVariance}
    \end{align}
    Both \(\mathbb{E}[Y^2 \mid Y > y_*]\) and \(\mathbb{E}[Y \mid Y > y_*]\) are computable using our framework as aforementioned. Whether \(Y\) is an individual loss or a sum of losses has no impact on the computational cost. However, when \(Y\) is a sum of losses, it naturally introduces the tail covariance.
		
\paragraph{The Tail covariance (TCov):}
Consider the sum \( S = Y_1 + Y_2 \) of two losses, we illustrate the problem using two losses but the results can be extended to higher dimension, and assume that the MGF $\Phi_{(Y_1,Y_2)}$ is known. As mentioned above, we can compute \(\TCE_S(s^*)\), and we can also decompose it into individual contributions since we have:
    \[
    \TCE_S(s^*) = \mathbb{E}[Y_1 | S > s_*] + \mathbb{E}[Y_2 | S > s_*].
    \]
    This decomposition helps in understanding the contribution of each component variable to the tail behavior of the sum. This quantity can be calculated using Proposition~\ref{prop:JointExpectedShortfallDerivativeFourierTransform}. Then considering the tail variance risk of the sum, it leads to a decomposition of the form:
    \begin{align}
    \TV_{S}(s_*)= \TV_{Y_1+Y_2}(s_*) = \TV_{Y_1|S}(s_*) + \TV_{Y_2|S}(s_*) + 2\TCov_{Y_1,Y_2|S}(s_*), \label{eq:tail_variance_sum}
    \end{align}
    where \(\TV_{Y_1|S}(s_*) = \mathbb{E}[(Y_1 - \mathbb{E}[Y_1 | S > s_*])^2 \mid S > s_*]=\mathbb{E}[Y_1^2 \mid S > s_*]- \mathbb{E}[Y_1 | S > s_*]^2 \), a similar equation for \(\TV_{Y_2|S}(s_*)\), and the tail covariance defined as		
		\begin{align}
     \TCov_{Y_1,Y_2|S}(s_*) = \mathbb{E}[Y_1Y_2 \mid S > s_*] - \mathbb{E}[Y_1 \mid S > s_*] \mathbb{E}[Y_2 \mid S > s_*], \label{eq:TCov}
    \end{align}
		 which provides insight into how dependence structures behave in the tail regions. Using our Proposition~\ref{prop:JointExpectedShortfallDerivativeFourierTransform} and Proposition~\ref{prop:TripleExpectedShortfallDerivativeFourierTransform}, all these conditional expectations can be exactly computed using one-dimensional integrations. Note that \(\TCov_{Y_1,S|S}(s_*)\) corresponds to the tail covariance risk measure of \citet[Eq.~(1.12)]{FurmanLandsman2006}.\\

		This showcases how joint distributions can yield decompositions of more complex tail risk measures, in particular for moments higher than two.

 \paragraph{The Tail skewness (TS):}  Following the definition of the tail variance, it is natural to define the tail skewness as:
    \begin{align}
    \TS_Y(y_*) = \frac{\mathbb{E}[(Y - \mathbb{E}[Y | Y > y_*])^3 \mid Y > y_*]}{\left(\mathbb{E}[(Y - \mathbb{E}[Y| Y > y_*])^2 \mid Y > y_*]\right)^{3/2}}. \label{eq:tail_skew}
    \end{align}
    It can similarly be computed once we have the higher conditional moments, all available via our integral transforms.\\
		
		If we consider the tail skewness of the sum \( S = Y_1 + Y_2 \) of two losses, then the denominator of \eqref{eq:tail_skew} is simply \(\TV_{S}(s_*)\) given by \eqref{eq:tail_variance_sum}, and straightforward computations show that the numerator of \eqref{eq:tail_skew} can be decomposed as
		\begin{align*}
    \mathbb{E}[(S - \mathbb{E}[S | S > s_*])^3 \mid S > s_*]&= \mathbb{E}[Y_1^3\mid S > s_*] +3\mathbb{E}[Y_1^2Y_2\mid S > s_*]+3\mathbb{E}[Y_1Y_2^2\mid S > s_*]\\
		&+ \mathbb{E}[Y_2^3\mid S > s_*]-3\TV_{S}(s_*)\TCE_{S}(s_*) + 2 \TCE_{S}(s_*)^3,		
    \end{align*}
    with these conditional higher joint moments of \(Y_1,Y_2\), that is, conditional expectations of the form \(\mathbb{E}[Y_1^p Y_2^q|S>s_*]\) with \(p,q \in \mathbb{N}\) and \( p+q=3\), computed thanks to our Proposition~\ref{prop:JointExpectedShortfallDerivativeFourierTransform} and Proposition~\ref{prop:TripleExpectedShortfallDerivativeFourierTransform}. This latter conditional expectation is named tail joint moment in \citet[Eq.~(1.4)]{YangWangYao2025} and is a building block for the  tail joint central moment in \citet[Eq.~(1.5)]{YangWangYao2025}. The tail skewness appears in \citet{YangWangYao2025} (where it is named tail conditional skewness (TCS)), this work also defines the tail conditional kurtosis that can be handled in our framework.\myfootnote{{\color{black}A worked-out implementation of the \(\TV\) for the GH distribution, computed using the density as in \citet{IgnatievaLandsman2025} and our approach, is provided in the supplementary appendix. This appendix also contains an implementation of \(\TS\) using our approach, since it has not yet been derived using the density.}}

\section{Numerical experiments}\label{sec:numerical_implementation}

In this section, we present a numerical implementation to demonstrate the effectiveness of our formulas. We begin by detailing the implementation steps that control the dependency between losses in our model. Since our approach relies on the MGF of the state variable rather than its density, we also briefly review the well-known results that allow us to recover the key quantities used in risk measurement.

\subsection{Implementation details}
In this section, we consider the two-dimensional case \(n=2\) and illustrate our results using the {\color{black}synthetic} parameters listed in Table~\ref{tab:model_param_values}.
\begin{center}
    [ Insert Table~\ref{tab:model_param_values} here ]
\end{center}

We set $\beta$ in accordance with the Bru case as defined by \citet{Bru1991} for computational efficiency. Additionally, we set the initial value \(x_0\) to be the mean value of the stationary solution of the stochastic differential equation \eqref{eq:Wishart}, it is given by the matrix \(\bar{x}_{\infty}\) which is the solution of \(-\beta \sigma^2 =  m \bar{x}_{\infty} + \bar{x}_{\infty} m^\top \). This solution is well defined since it is a Lyapunov equation and we assume that the eigenvalues of \(m\) have negative real parts (\(\beta \sigma^2\) is a symmetric definite positive matrix by construction).\\  

From \eqref{eq:Wishart}, we can deduce the instantaneous covariance structure between the components of the Wishart process as follows: 
\begin{align} d\langle x_{11,.}, x_{11,.}\rangle_t &= 4 x_{11,t} (\sigma_{11}^2 + \sigma_{12}^2)dt, \label{eq:bracket_x11x11}\\ 
d\langle x_{22,.}, x_{22,.}\rangle_t &= 4 x_{22,t} (\sigma_{12}^2 + \sigma_{22}^2)dt, \label{eq:bracket_x22x22}\\ 
d\langle x_{12,.}, x_{12,.}\rangle_t &= x_{11,t}(\sigma_{12}^2 + \sigma_{22}^2)dt + 2x_{12,t}\sigma_{12}(\sigma_{11} + \sigma_{22})dt + x_{22,t}(\sigma_{11}^2 + \sigma_{12}^2)dt,\label{eq:bracket_x12x12}\\ 
d\langle x_{11,.}, x_{22,.} \rangle_t &= 4x_{12,t}\sigma_{12}(\sigma_{11} + \sigma_{22})dt. \label{eq:bracket_x11x22} 
\end{align}

If \( x_{11} \) and \( x_{22} \) represent two losses, their dependence is controlled by \eqref{eq:bracket_x11x22}. Setting \( \sigma_{12} = 0 \) results in a null instantaneous covariation. However, \eqref{eq:bracket_x12x12} shows that \( d\langle x_{12,.}, x_{12,.}\rangle_t \neq 0 \) even if \( \sigma_{12} = 0 \), since \( \sigma \in \mathbb{S}_2^{++} \) implies \( \sigma_{11} > 0 \) and \( \sigma_{22} > 0 \). The sign of this dependence is determined by the sign of \( \sigma_{12} \), which can be positive or negative as \( \sigma \in \mathbb{S}_2^{++} \). If \( \sigma_{12} = 0 \), then \( \sigma^2 \) is a diagonal matrix (\textit{i.e.}, \( (\sigma^2)_{12} = 0 \)). Note that if we wish to change to zero the dependence between \(x_{11,t}\) and \(x_{22,t}\), setting \(\sigma_{12}=0\) also inadvertently lowers the quadratic variation of \(x_{11,t}\) and \(x_{22,t}\) (given by \eqref{eq:bracket_x11x11} and \eqref{eq:bracket_x22x22}). Thus, to counteract this effect, we introduce a new matrix called $\tilde{\sigma}$ defined by
\begin{equation}\label{eq:sigma_tilde}
\tilde{\sigma}=
\begin{pmatrix}
\tilde{\sigma}_{11} & \tilde{\sigma}_{12} \\
\tilde{\sigma}_{12} & \tilde{\sigma}_{22} 
\end{pmatrix}
=
\begin{pmatrix}
\sqrt{\sigma_{11}^2 + \sigma_{12}^2} 	& 0 \\
0 																		& \sqrt{\sigma_{22}^2 + \sigma_{12}^2}
\end{pmatrix},
\end{equation}
and define a Wishart process \(\tilde{x}_t\) with dynamic given by \eqref{eq:Wishart} but with \(\sigma\) replaced with \(\tilde{\sigma}\). We call this process the zero-dependent equivalent (to $x_t$) process since, by construction, we have 
\begin{align}
\frac{d\langle \tilde{x}_{11,.}, \tilde{x}_{11,.}\rangle_t}{\tilde{x}_{11,t}} &= \sigma_{11}^2 + \sigma_{12}^2 =\frac{d\langle x_{11,.}, x_{11,.}\rangle_t}{x_{11,t}},\label{eq:tilde_sigma11_equal_sigma211}\\
d\langle \tilde{x}_{11,.}, \tilde{x}_{22,.} \rangle_t &=0, \nonumber \\
\frac{d\langle \tilde{x}_{22,.}, \tilde{x}_{22,.}\rangle_t}{\tilde{x}_{22,t}} &= \sigma_{12}^2 + \sigma_{22}^2 =\frac{d\langle x_{22,.}, x_{22,.}\rangle_t}{x_{22,t}}. \nonumber
\end{align}
The process \( \tilde{x}_t \) allows us to quantify the impact of the dependence between \( x_{11,t}\) and \( x_{22,t}\) by changing the value of $\sigma_{12}$ but keeping the quadratic variations unchanged.\\

In many applications, we need to find \( x_* \) for a given (scalar) variable \( X \) and \(\alpha \in ]0,\;1[\) such that \( \mathbb{E}[\mathbf{1}_{\{ X > x_* \}}] = \alpha \). In these applications, \( X \) is defined as \( X = \mathrm{tr}[\theta x_t] \) with \( \theta \in \mathbb{S}_n^{+} \), making \( X \) a positive random variable. The Laplace transform of \( X \) is given by \( \mathcal{L}_X(z) = \mathbb{E}[e^{-z X}] = \Phi(t, -z \theta, x_0) \) for \( z > 0 \), where \( \Phi(\cdot, \cdot, \cdot) \) is defined by \eqref{eq:mgf}. Using the inverse Laplace transform, the cumulative distribution function of \( X \) is:
$$
F_X(u) = \frac{1}{\pi \mathrm{i}} \int_{\bar{\gamma} - \mathrm{i}0}^{\bar{\gamma} + \mathrm{i}\infty} \Re\left( \frac{e^{u z} \mathcal{L}_X(z)}{z} \right) dz,
$$
with $\bar{\gamma}>0$ and the integral computed using a quadratic integration algorithm. Since the function \( u \mapsto F_X(u) \) is monotonic, the value of \( x_* \) such that \( \bar{F}_X(x_*) = 1 - F_X(x_*) = \mathbb{E}[\mathbf{1}_{\{ X > x_* \}}] = \alpha \) is found iteratively.\\

We use the following variables to illustrate our approach. First, we consider \(x_{11,t}=\mathrm{tr}[\theta x_{t}]\) with \( \theta =e_{11} \), and since \(\theta  \in \mathbb{S}_2^{+}\) we use the Laplace/inverse Laplace transforms above to find the truncating value for the distribution of \( x_{11,t}\). We take \(x_*=1.0\) and for which \(\mathbb{E}[\mathbf{1}_{\lbrace x_{11,t}>x_*\rbrace }]=0.0910 \) when \( t= 1\) and the parameters of Table \ref{tab:model_param_values}.\\

As we are interested in a portfolio of losses, we consider \(s_t = x_{11,t}+x_{22,t}=\mathrm{tr}[\theta x_t]\) with \(\theta = I_2\in \mathbb{S}_2^{++}\), so that \(s_t\) is a positive random variable with a distribution that we truncate at the level \( s_* = 1.3\) for which  \(\mathbb{E}[\mathbf{1}_{\lbrace s_{t}>s_*\rbrace }]=0.0584 \) when \( t= 1\) and the parameters of Table \ref{tab:model_param_values}. This value is also obtained using the Laplace/inverse Laplace transforms above.\\

Lastly, the equation \eqref{eq:bracket_x11x22} implies that the dependence between \(x_{11,t}\) and \(x_{22,t}\) is stochastic and its distribution is given by \(x_{12,t}\) and it can be used to perform the conditioning of a risk measure. Writing \(x_{12,t}=\mathrm{tr}[\theta x_t]\) with \(\theta = e_{12}/2 +e_{21}/2\) and choosing \(x_{12,*}=0.435\), it leads to \(\mathbb{E}[\mathbf{1}_{\lbrace x_{12,t}>x_{12,*}\rbrace }]=0.0544 \) for \(t=1\) and the parameters of Table \ref{tab:model_param_values}. Note that as \(x_{12,t}\) does not need to have a positive distribution, the Laplace/inverse Laplace transform above has to be replaced by a Fourier/inverse Fourier transform but since the procedure is similar, we omit the details.

\subsection{One date risk measures}
In this subsection, we present and analyze various risk measures computed for the model parameters specified in Table~\ref{tab:model_param_values} and the thresholds specified above. The values are reported in Table~\ref{tab:risk_measures}.

\begin{center}
    [ Insert Table~\ref{tab:risk_measures} here ]
\end{center}

The first four values provide the first two moments of \(x_{11,t}\) and \(s_t\) conditional on the events \(x_{11,t} > x_*\) and \(s_t > s_*\), respectively. The next four conditional moments, \(\mathbb{E}[x_{11,t} \mid s_t > s_*]\), \(\mathbb{E}[x_{11,t}^2 \mid s_t > s_*]\), \(\mathbb{E}[s_t \mid x_{11,t} > x_*]\), and \(\mathbb{E}[s_t^2 \mid x_{11,t} > x_*]\), illustrate the impact of changing the conditioning set. This involves considering the event \(s_t > s_*\) when computing the first two moments of \(x_{11,t}\), or the event \(x_{11,t} > x_*\) when computing the first two moments of \(s_t\). The values suggest that changing the conditioning event affects the second moment more than the first moment. Specifically, the first (conditional) moment of \(x_{11,t}\) changes by \(1.83\%\), while its second moment changes by 3.74\%. For \(s_t\), the corresponding changes are \(-3.07\%\) and \(-6.12\%\), respectively. The framework also allows us to use the conditioning event \(x_{12,t} > x_{12,*}\), capturing the effects of the covariance between the components of \(x_t\) when computing the first two (or higher) conditional moments of the aggregated losses \(s_t\). When considering the first two moments of \(s_t\), conditioning on this event leads to values closer to those obtained when conditioning on \(s_t > s_*\) than when conditioning on \(x_{11,t} > x_*\). The last value, \(\mathbb{E}[x_{11,t} s_t \mid s_t > s_*]\), illustrates the ability to compute conditional cross-moments. Figure \ref{fig:expectation_x_star} illustrates the sensitivity \(x_* \to \mathbb{E}\left[ x_{11,t} |  x_{11,t}>x_* \right]\) while  Figure \ref{fig:expectation_s_star} reports \( s_* \to \mathbb{E}\left[ s_{t} | s_{t}>s_* \right]\).

\begin{center}
    [ Insert Figure~\ref{fig:expectation_x_star} here ]
\end{center}

\begin{center}
    [ Insert Figure~\ref{fig:expectation_s_star} here ]
\end{center}

\subsection{Effect of correlation on risk measures}\label{sec:correl_measures}
 We analyze the effect of changing the dependence between the losses \(x_{11,t}\) and \(x_{22,t}\) on different risk measures. As explained earlier, this is achieved by setting the off-diagonal term of \(\sigma\) to zero. We denote by \(\tilde{x}_t\) the zero-dependent equivalent process to \(x_t\) and compute the different risk measures, keeping \(x_*\) and \(s_*\) unchanged. The results of the tail conditional expectation calculations are reported in Table~\ref{tab:risk_measures_sigma}.

\begin{center}
    [ Insert Table~\ref{tab:risk_measures_sigma} here ]
\end{center}

Analyzing the effect on shortfall measures with \(\sigma\)'s off-diagonals set to zero reveals interesting findings. Conditioning on \(x_{11,t}\) does not change the risk measure for both the first and second moments relative to the unchanged \(\sigma\) case, as expected since \(\tilde{\sigma}\) satisfies \eqref{eq:tilde_sigma11_equal_sigma211}. However, analyzing the sum \(\tilde{s}_t\) shows a decrease in shortfall compared to the correlated case, with changes of \(-1.71\%\) and \(-4.84\%\) for \(\mathbb{E}[ \tilde{s}_t \mid \tilde{s}_t > s_* ]\) and \(\mathbb{E}[ \tilde{s}_t^2 \mid \tilde{s}_t > s_* ]\), respectively. This indicates a narrowing of the tails due to the removal of dependence between the losses. The narrowing also explains the changes observed in \(\mathbb{E}[ \tilde{s}_t \mid \tilde{x}_{11,t} > x_* ]\) and \(\mathbb{E}[ \tilde{s}_t^2 \mid \tilde{x}_{11,t} > x_* ]\), with the values being \(-5.07\%\) and \(-13.37\%\), respectively. Conversely, \(x_{11,t}\) conditioned on \(\tilde{s}_t\) being above \(s_*\) increases relative to the dependent variant due to the lower expected level of \(\tilde{s}_t\) caused by the narrower tail. This also explains why the cross-moment \( x_{11,t} s_t \) conditioned on \(s_t > s_* \) increases by \(1.11\%\). Practically, it is worth noting that higher-order moments are more affected by changes in the dependence structure between the losses. Since the dependence between losses can be difficult to measure, setting its value to zero can significantly alter risk profiles and the distribution of outcomes.\\

\subsection{Two dates risk measures}
 We conduct numerical experiments on two different dates, \( t_1 \) and \( t_0 \) with \( t_0 < t_1 \), to compute the expected shortfall based on Proposition~\ref{prop:ConditionalMomentsOfXTwoDates2}. This proposition allows us to compute future expected moments conditional on past losses exceeding a certain threshold. For the dates, we take \(t_0 = 1\), as in the previous examples, while \(t_1 = 1.5\). The thresholds are those of the previous examples and are therefore consistent with the distribution of the process at that date. Table~\ref{tab:risk_measures_two_dates} reports the first and second moments of \(x_{11,t_1}\) conditional on \(x_{11,t_0} > x_*\) and similar quantities for \(s_{t_1}\) conditioned on \(s_{t_0} > s_*\), as well as the difference (in \(\%\)) with the corresponding values of Table~\ref{tab:risk_measures}. All the differences are negative, a result that is expected since the process reverts to its long-term mean, which is lower than its conditional expectation at earlier times. Note also that higher-order moments are more affected by this mean-reverting behavior. This modeling aspect leverages the time-dependent nature of the process and allows us to develop an intertemporal {\color{black}or time-lagged} point of view of risk measures. \\

To further illustrate this feature, Figure~\ref{fig:t1_vs_output_plot} displays the exponential decay of the two-date expected shortfall \(\mathbb{E}[x_{11,t_1} \mid x_{11,t_0} > x_*]\) over the range \(t_1 \in [1, 200]\) to its long-term mean value of \(0.84\), with substantial decay occurring only around \(t = 200\) and that is controlled by the matrix \(m\), in particular \(m_{11}\). This mean-reverting behavior can also be combined with a different conditioning, and to this end, the table also reports the first and second moments of \(x_{11,t_1}\) conditioned on \(s_{t_0} > s_*\) and similar quantities for \(s_{t_1}\) conditioned on \(x_{11,t_0} > x_*\), as well as the difference (in \(\%\)) with the corresponding values of Table~\ref{tab:risk_measures}. All the conditional expectations are smaller than their counterparts, and of course, the mean-reverting behavior is the reason for such a result.

\begin{center}
    [ Insert Table~\ref{tab:risk_measures_two_dates} here ]
\end{center}

\begin{center}
    [ Insert Figure~\ref{fig:t1_vs_output_plot} here ]
\end{center}

\subsection{Effect of dependency in capital allocation}\label{sec:XuMaoCapitalAllocation}
Another use of the conditional higher moments is to solve the capital allocation problem based on the tail mean-variance model defined in \citet{XuMao2013}. The problem involves optimally allocating capital between losses to minimize the mean-variance of the sum of the square of the losses, conditioned on a certain level of portfolio risk. Mathematically, Theorem 2.1 of \citet{XuMao2013} provides an explicit solution for a random vector of losses \((Y_1, \ldots, Y_{n})\) to the allocation problem below:
\begin{align}
	\min_{\mathbf{p} \in A} \left\{ \mathbb{E}\left[\sum_{i=1}^n (Y_i - p_i)^2 \mid Z>z_* \right] + \gamma \text{Var} \left(\sum_{i=1}^n (Y_i - p_i)^2 \mid Z>z_* \right) \right\}, \label{eq:XuMao2013_optimisation_problem}
\end{align}
where \(A = \{\mathbf{p} \in \mathbb{R}^n : p_1 + \ldots + p_n = c \}\) with the allocation budget \(c > 0\), \(\gamma > 0\), and \(Z\) is the conditioning variable and \(z_*\) a given threshold. The above conditional expectations can be reformulated in our framework using Remark \ref{rem:nExpectedShortfallDerivativeFourierTransform}. \citet{XuMao2013} illustrate their results using the multivariate elliptical distribution. \citet{LiYin2024} analyze the same model using the skew-elliptical distribution but underline the difficulty of the problem. Indeed, they state, ``... the computational challenge arises due to the calculation involving conditional joint moments of underlying risks" (see \citet[p. 2]{LiYin2024}). \citet{YangWangYao2025} solve the same problem for the multivariate generalized hyperbolic distribution. One of the difficulties with this model is that it necessitates the computation of \(\mathbb{E}[Y_j^2Y_i|Z>z_*]\) and \(\mathbb{E}[Y_jY_i|Z>z_*]\) for \(j,i=1,\ldots,n\). Thanks to the strong analytical properties of our results, this difficulty is easily handled in our framework. Note that \(\mathbb{E}[Y_jY_i|Z>z_*]\) is the first term of \(\TCov_{Y_1,Y_2|Z}(z_*)\) of \eqref{eq:TCov} {\color{black} while \(\mathbb{E}[Y_j^2Y_i|Z>z_*]\) is a third-order conditional moment}.\\

Note that the original result of \citet{XuMao2013} is based on conditioning on the sum of the random variables \(Y_i\) (\textit{i.e.}, setting \(Z\) to be \(S = \sum_{i=1}^{n} Y_i\) with \(S > \text{VaR}_q(S)\), the conditioning of the sum on the value-at-risk). However, the proof of their result does not require this specific condition, and thus we can choose \(Z\) to be another variable or combination of variables. Since this problem heavily relies on the dependence between the losses, the conditioning is performed on the variable \(x_{12,t}\), thereby quantifying the impact of loss dependence on the capital allocation. It exploits the matrix nature of the Wishart process for modeling risk measures. It will demonstrate the importance of appropriately accounting for dependency in risk management within the Wishart process. To this end, we proceed as above, solving the problem using the parameters of Table~\ref{tab:model_param_values} and then considering the process \(\tilde{x}_t\), the zero-dependent equivalent process to \(x_t\), and solving the problem \eqref{eq:XuMao2013_optimisation_problem}. A comparison between the two solutions illustrates the impact of the losses' dependence on the capital allocation.\\

First, we consider the capital allocation problem with Wishart parameters in Table~\ref{tab:model_param_values}, time value of \(t = 1\), the allocation budget \(c = 1.3\), \(\gamma = 1\), and set the conditioning random variable to be \(x_{12,t} > \text{VaR}_{0.95}(\text{tr}[\theta x_t])\), where \(\theta = (e_{12} + e_{21})/2\). Note that the values of \(\text{VaR}_{0.95}(\text{tr}[\theta x_t])\) and \(\text{VaR}_{0.95}(\text{tr}[\theta \tilde{x}_t])\) are different to account for the difference in the quadratic variations between \(x_{11,t}\) and \(x_{22,t}\) and between \(\tilde{x}_{11,t}\) and \(\tilde{x}_{22,t}\). In particular, for the dependent case, we have \(\text{VaR}_{0.95}(\text{tr}[\theta x_t]) = 0.438\), while in the diagonal case \(\tilde{\sigma}\) given by \eqref{eq:sigma_tilde} we have \(\text{VaR}_{0.95}(\text{tr}[\theta \tilde{x}_t]) = 0.085\).\\

The optimal allocation values for \eqref{eq:XuMao2013_optimisation_problem} for the case in Table~\ref{tab:model_param_values} are \(p_{1} = 1.031\) and \(p_{2} = 0.269\), which gives a ratio of \(p_{1}/p_{2} = 3.836\). In comparison, when using \(\tilde{x}_t\), the optimal allocation is \(\tilde{p}_{1} = 0.965\) and \(\tilde{p}_{2} = 0.335\), with a ratio of \(\tilde{p}_{1}/\tilde{p}_{2} = 2.877\). The difference in the allocations is best viewed through their ratios, and the impact of the dependence on this ratio is substantial. Having a dependency between the losses suggests allocating more to the loss with more variance (in this case, \(x_{11,t}\)). Thus, an erroneous assumption of independence or a significant error in its estimation can lead to a situation where one loss is under-capitalized and another over-capitalized.\\

As previously demonstrated, the Wishart process is highly tractable and is a natural framework to model dependent losses since it is, by construction, a positive definite matrix process. Furthermore, due to its tractability, we are even able to derive allocation results, and the results highlight that if dependence is not appropriately accounted for, it can lead to incorrect allocation of resources.

{\color{black}\subsection{Empirical estimation}\label{sec:empirical_estimation}
Practical applicability is an important criterion when evaluating the usefulness of a model. We illustrate our approach using the Danish Fire Loss dataset, a standard benchmark in insurance. The raw data contain 2,167 fire insurance claims from 1980 to 1990, measured in millions of DKK.\footnote{Note that we do not apply a logarithmic transformation, as in \citet{IgnatievaLandsman2021}.} We focus on the \emph{Building} and \emph{Contents} lines, excluding \emph{Profit} due to the many zeros, and we aggregate claim amounts to weekly totals, retaining weeks with strictly positive totals in both lines. Throughout, we treat these weekly bivariate observations as i.i.d. draws rather than a time series. As a result, for the model we estimate the parameters of its limiting, stationary distribution. By Corollary~\ref{coro:MGFInfinity}, this distribution is matrix gamma with scale $2\varsigma_{\infty}$ and shape  $\beta/2$. The pair $(\beta,\varsigma_{\infty})$ is obtained by the method of moments, as in \citet{GourierouxJasiakSufana2009} or \citet[Proposition~2]{GouerierouxLu2025}, and the explicit equations are given in the supplementary appendix. In other words, we fit the stationary law to the weekly aggregates, then compute the risk measures developed in this paper from the fitted law.\\

If only the diagonal terms of the process are used to estimate the parameters, then only \(\varsigma_{ij,\infty}^2\) can be identified; see the supplementary appendix for the expressions of the estimators. Since we have \(\mathbb{E}[x_{ij,\infty}]  =  \beta  \varsigma_{ij,\infty}\), this suggests specifying the sign of \( \varsigma_{ij,\infty}\) using the sign of the correlation between \(x_{ii,\infty}\) and \(x_{jj,\infty}\). These estimators yield the model parameters reported in Table~\ref{tab:gl_mom_estimates}.

\begin{center}
    [ Insert Table~\ref{tab:gl_mom_estimates} here ]
\end{center}

The implied correlation from the fitted matrix gamma distribution is 0.44, and is close to the sample correlation of 0.56 from the weekly aggregates, which indicates a reasonable fit. Using Proposition~\ref{prop:ExpectedShortfallFourierTransform}, we compute the model implied risk measures, see Table~\ref{tab:matrix_gamma_risk_measures}. For reference, we also report the empirical, nonparametric counterparts computed directly from the weekly data. The model implied quantities track their empirical analogues closely, which supports the adequacy of the limiting distribution specification for this application.

\begin{center}
    [ Insert Table~\ref{tab:matrix_gamma_risk_measures} here ]
\end{center}
}

\section{Conclusion}\label{Conclusion}

In this work, we introduce an analytical framework capable of quantifying multivariate risk measures while incorporating temporal aspects. The flexibility of the framework lies in its reliance solely on the characteristic function. We present Fourier transforms of tail conditional moments and their cross-moments, and show that this approach reduces a two-dimensional integration problem to a one-dimensional one. We compute the tail conditional expectation for various cross-moments, both at a single point in time and across multiple time points. This flexible framework is demonstrated using the Wishart process---a highly tractable matrix affine process characterized by its positive definiteness---making it well-suited for loss modeling. We provide a numerical implementation of the derived framework and demonstrate the impact of dependency between loss items. Our results show that removing dependency can lead to significant differences in the tail conditional expectation, with variations of up to 13.37\% at a fixed point in time and up to 18.86\% across time. We further demonstrate that our method and model can be applied in the context of capital allocation, showcasing that we are able to obtain operational results for the model proposed by \citet{XuMao2013}. Overall, our framework offers a robust and flexible tool for risk management, capable of handling complex dependencies and temporal evolutions in multivariate loss processes.

\clearpage

\footnotesize
\linespread{0}
\selectfont
\bibliographystyle{abbrvnat}
\bibliography{Biblio}
\clearpage

\appendix
\clearpage
\footnotesize
\linespread{1.1}
\selectfont
\section*{Tables and figures}

\subsection*{Tables}

\begin{table}[htbp]
	\caption{Wishart model parameter values.}
	\label{tab:model_param_values}
	\begin{center}
		\begin{tabular}{c c}
			Variable & Value \\
			\hline
			$\beta$ & 4 \\
			$x_0$ & $\begin{pmatrix} 0.84 & x_{12} \\ x_{12} & 0.22 \end{pmatrix}$ \\
			$x_{12}$ & $\approx 0.32$ \\
			$m$ & $\begin{pmatrix} -0.01 & 0 \\ 0 & -0.02 \end{pmatrix}$ \\
			$\sigma$ & $\begin{pmatrix} 0.06 & \sigma_{12} \\ \sigma_{12} & 0.04 \end{pmatrix}$ \\
			$\sigma_{12}$ & $\rho_{\sigma}  \sqrt{0.06 \times 0.04}$ \\
			$\rho_{\sigma}$ & 0.5 \\
			\hline
		\end{tabular}
	\end{center}
	{\footnotesize \textit{Note.} This table contains the model parameters used for our numerical experiments in Section~\ref{sec:numerical_implementation}. For the specified parameters $\beta$, $m$ and $\sigma$, the initial value of the process $x_0$ is {\color{black}the numerical} solution of $-\beta \sigma^2 = x_0 m +  x_0 m^\top$. }
\end{table}

\begin{table}
\begin{center}
\caption{Risk measures.}
\label{tab:risk_measures}
\[
\begin{array}{lc}
\textrm{Risk measure} & \textrm{Value} \\ \hline
\mathbb{E}\left[ x_{11,t} |  x_{11,t}>x_* \right]  			& 1.0613 \\
\mathbb{E}\left[ x_{11,t}^2 | x_{11,t}>x_* \right]  		& 1.1293	\\
\mathbb{E}\left[ s_{t} | s_{t}>s_* \right] 							& 1.3729	\\
\mathbb{E}\left[ s_{t}^2 | s_{t}>s_* \right] 						& 1.8892	\\
\mathbb{E}\left[ x_{11,t} |  s_{t}>s_* \right] 					& 1.0807	\\
\mathbb{E}\left[ x_{11,t}^2 | s_{t}>s_* \right] 				& 1.1715 	\\
\mathbb{E}\left[ s_{t} |  x_{11,t}>x_* \right] 					& 1.3320	\\
\mathbb{E}\left[ s_{t}^2 |  x_{11,t}>x_* \right] 				& 1.7803 	\\
\mathbb{E}\left[ s_{t} |  x_{12,t}>x_{12,*} \right] 	& 1.3628	\\
\mathbb{E}\left[ s_{t}^2 |  x_{12,t}>x_{12,*} \right] & 1.8635 	\\
\mathbb{E}\left[ x_{11,t} s_{t} | s_{t}>s_* \right] 		& 1.4871 	\\ \hline
\end{array}
\]

\end{center}
{\footnotesize \textit{Note.} This table contains the risk measures for the parameter values of Table \ref{tab:model_param_values} and $t=1$, $x_*=1$, $s_*=1.3$ and $x_{12,*}=0.435$.}
\end{table}

\begin{table}
\begin{center}
\caption{Risk measures without dependence.}\label{tab:risk_measures_sigma}
\[
\begin{array}{lcc}
\textrm{Risk measure} & \textrm{Value} & \textrm{Diff.} \%																						\\ \hline
\mathbb{E}\left[ \tilde{x}_{11,t} |  \tilde{x}_{11,t}>x_* \right]  			& 1.0613 &  0.00 	\\
\mathbb{E}\left[ \tilde{x}_{11,t}^2 |  \tilde{x}_{11,t}>x_* \right]  		& 1.1293 &  0.00	\\
\mathbb{E}\left[ \tilde{s}_{t} |  \tilde{s}_{t}>s_* \right] 							& 1.3558 & -1.71	\\
\mathbb{E}\left[ \tilde{s}_{t}^2 |  \tilde{s}_{t}>s_* \right] 						& 1.8408 & -4.84	\\
\mathbb{E}\left[ \tilde{x}_{11,t} |  \tilde{s}_{t}>s_* \right] 					& 1.1033 &  2.25  \\
\mathbb{E}\left[ \tilde{x}_{11,t}^2 |  \tilde{s}_{t}>s_* \right] 				& 1.2214 &  4.99  \\
\mathbb{E}\left[ \tilde{s}_{t} |  \tilde{x}_{11,t}>x_* \right] 					& 1.2813 & -5.07	 \\
\mathbb{E}\left[ \tilde{s}_{t}^2 |  \tilde{x}_{11,t}>x_* \right] 				& 1.6466 & -13.37 \\
\mathbb{E}\left[ \tilde{x}_{11,t} \tilde{s}_{t} |  \tilde{s}_{t}>s_* \right] 		& 1.4982 &	1.11   \\ \hline
\end{array}
\]

\end{center}
{\footnotesize \textit{Note.} This table contains the risk measures for the zero-equivalent process \(\tilde{x}\), as defined in section \ref{sec:numerical_implementation}, as well as the difference (in \%) with the values of Table \ref{tab:risk_measures} when the parameter values are those of Table \ref{tab:model_param_values} but with $\sigma$ replaced with \( \tilde{\sigma}\) given by \eqref{eq:sigma_tilde} and $t=1$, $x_*=1$ and $s_*=1.3$.}
\end{table}

\begin{table}[h!]
\begin{center}
\caption{Risk measures (two dates).}
\label{tab:risk_measures_two_dates}
\[
\begin{array}{lcc}
\textrm{Risk measure} & \textrm{Value} & \textrm{Diff.} \% \\ \hline
\mathbb{E}\left[ x_{11,t_1} \mid  x_{11,\,t_0}>x_* \right]  	&  1.0405 & -2.00 \\
\mathbb{E}\left[ x_{11,t_1}^2 \mid  x_{11,\,t_0}>x_* \right]  &  1.0944 & -3.19 \\
\mathbb{E}\left[ s_{t_1} \mid  s_{t_0}>s_* \right] 					&  1.3406 & -2.41 \\
\mathbb{E}\left[ s_{t_1}^2 \mid  s_{t_0}>s_* \right] 				&  1.8160 & -4.03 \\
\mathbb{E}\left[ x_{11,t} |  s_{t}>s_* \right] 							& 1.0594	& -1.97 \\
\mathbb{E}\left[ x_{11,t}^2 | s_{t}>s_* \right] 						& 1.1346 	&	-3.14	\\
\mathbb{E}\left[ s_{t} |  x_{11,t}>x_* \right] 							& 1.3011	&	-2.32 \\
\mathbb{E}\left[ s_{t}^2 |  x_{11,t}>x_* \right] 						& 1.7127 	& -3.80 \\ 
 \hline
\end{array}
\]
\end{center}
{\footnotesize \textit{Note.} This table contains the risk measures for the parameter values of Table \ref{tab:model_param_values}, as well as the difference (in \(\%\)) with the corresponding values of Table \ref{tab:risk_measures} and $t_0=1$, $t_1=1.5$, $x_*=1$ and $s_*=1.3$.}
\end{table}

\begin{table}[htbp]
	\caption{Danish Fire Loss parameter estimates.}\label{tab:gl_mom_estimates}
  \begin{center}
    \begin{tabular}{c c}
      \hline
      Variable & Value \\
      \hline
      $\beta$ & $3.24$ \\
      $\varsigma_{\infty}$ & $\begin{pmatrix} 7.09 & 4.65 \\ 4.65 & 9.60 \end{pmatrix}$ \\
      \hline
    \end{tabular}
  \end{center}
  {\footnotesize \textit{Note.} This table contains the method-of-moments estimates of the limiting distribution parameters from \citet[Proposition 2]{GouerierouxLu2025} for Corollary~\ref{coro:MGFInfinity}: $\beta$ and $\varsigma_{\infty}$ for the aggregated weekly losses for the Danish Fire Loss dataset on the \emph{Building} and \emph{Contents} lines of business. See the supplementary appendix for the equations used to compute these estimates.}
\end{table}

\begin{table}
\begin{center}
\caption{Risk measures of matrix gamma.}
\label{tab:matrix_gamma_risk_measures}
\[
\begin{array}{lcc}
\textrm{Risk measure} & \textrm{Model} & \textrm{Empirical} \\ \hline
\mathbb{E}\!\left[ x_{11,\infty} \mid x_{11,\infty} > x_* \right] & 74.09 & 88.60 \\
\mathbb{E}\!\left[ x_{11,\infty}^2 \mid x_{11,\infty} > x_* \right] & 5732.04 & 10541.72 \\
\mathbb{E}\!\left[ x_{22,\infty} \mid x_{22,\infty} > x_* \right] & 100.40 & 82.66 \\
\mathbb{E}\!\left[ x_{22,\infty}^2 \mid x_{22,\infty} > x_* \right] & 10534.68 & 8749.95 \\
\mathbb{E}\!\left[ s_{\infty} \mid s_{\infty} > s_* \right] & 145.47 & 156.27 \\
\mathbb{E}\!\left[ s_{\infty}^2 \mid s_{\infty} > s_* \right] & 21632.63 & 28409.03 \\ \hline
\end{array}
\]
\end{center}
{\footnotesize \textit{Note.} Entries report tail conditional expectations at quantile $0.95$.
For $i=1,2$, $x_{ii,\infty}$ denotes the limiting diagonal components and $s_\infty$ the sum; thresholds $x_*, s_*$ are the corresponding VaR levels (model VaR for the model column, empirical quantile for the empirical column).}
\end{table}

\clearpage

\subsection*{Figures}

\begin{figure}[h!]
    \begin{center}
		
    \caption{Dependence of \(\mathbb{E}\left[ x_{11,t} \mid  x_{11,t}>x_* \right]\) on \(x_*\).}
		\label{fig:expectation_x_star}
		\includegraphics[width=0.6\textwidth]{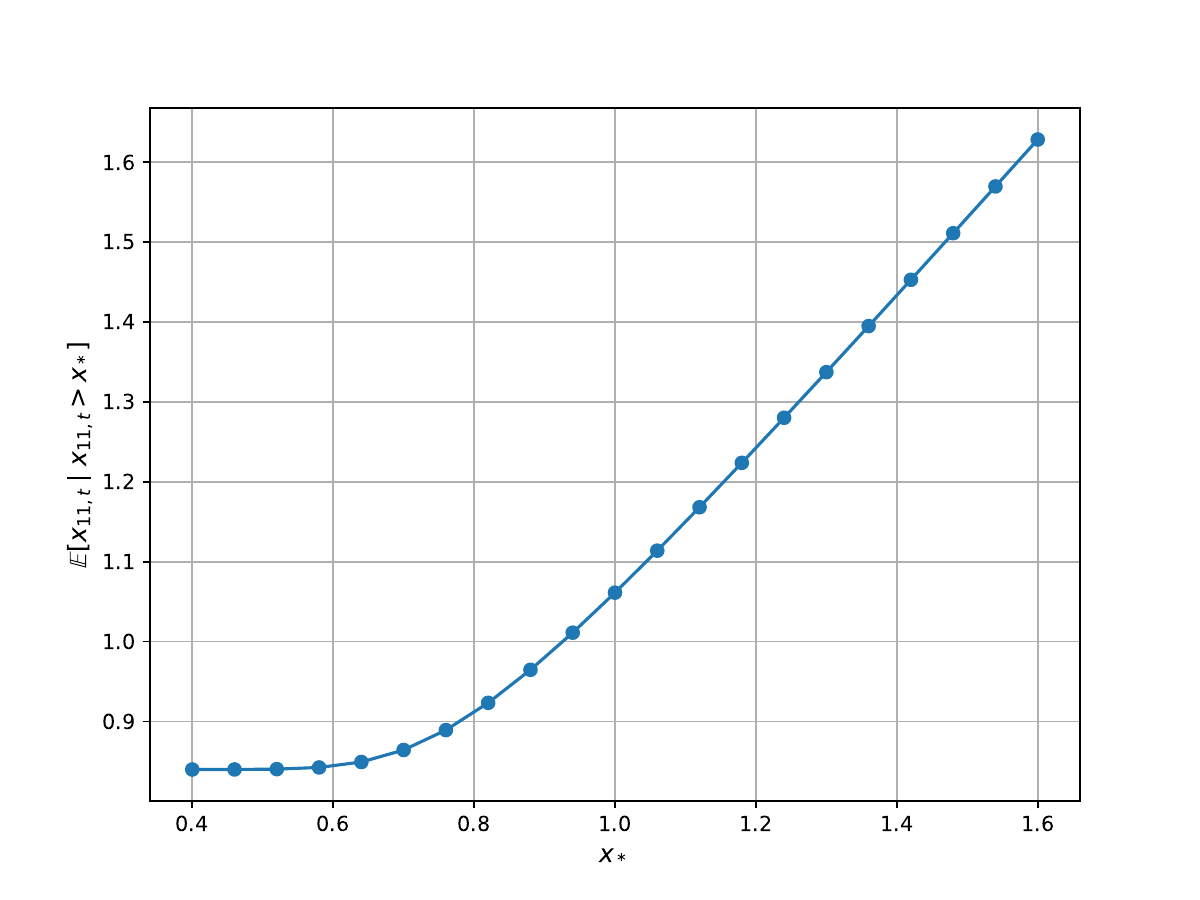}
    \end{center}
		Note: Sensitivity of \(\mathbb{E}\left[ x_{11,t} \mid  x_{11,t}>x_* \right]\) on \(x_*\) for \(t = 1\).
\end{figure}

\begin{figure}[h!]
    \begin{center}
		
    \caption{Dependence of \(\mathbb{E}\left[ s_{t} \mid  s_{t}>s_* \right]\) on \(s_*\).}
		\label{fig:expectation_s_star} 
    \includegraphics[width=0.6\textwidth]{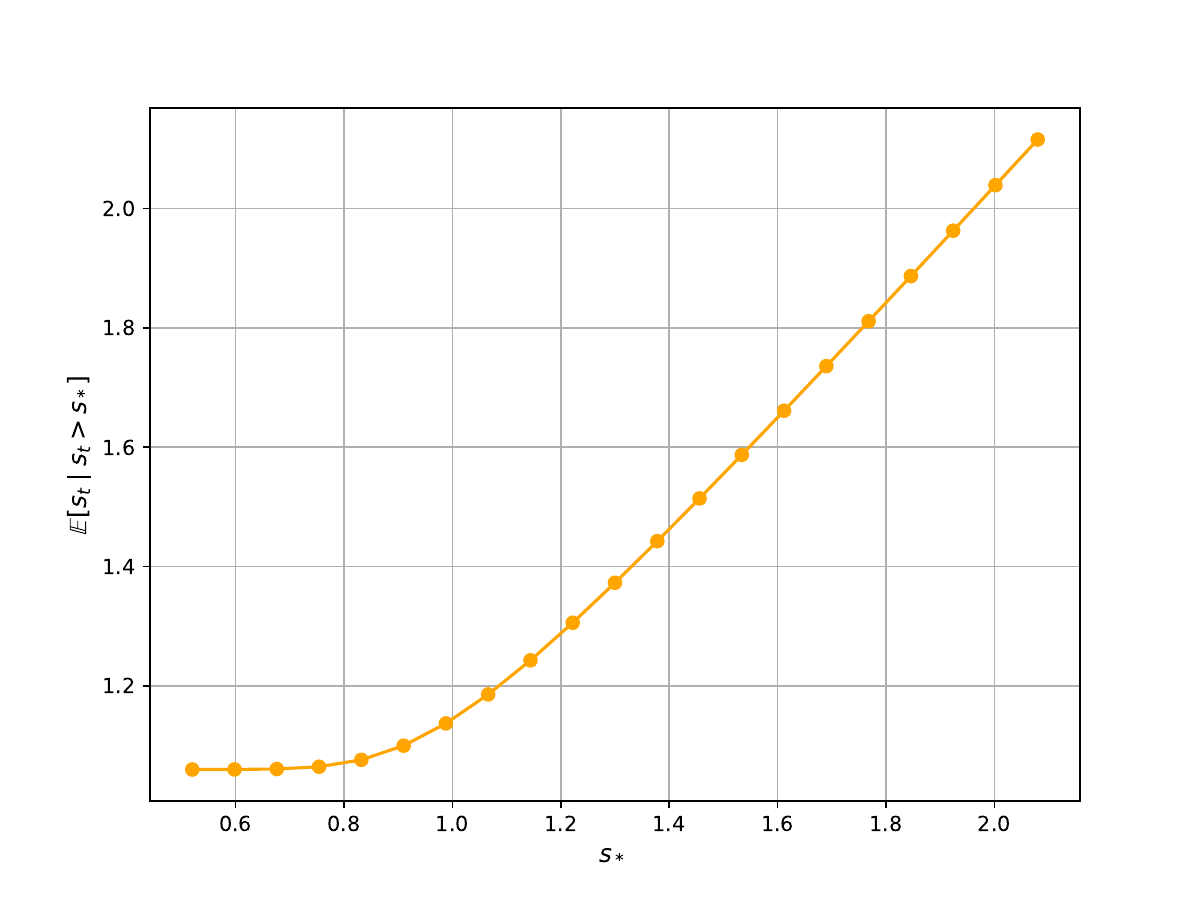}
    \end{center}
		Note: Sensitivity of \(\mathbb{E}\left[ s_{t} \mid  s_{t}>s_* \right]\) on \(s_*\) for \(t = 1\).
\end{figure}

\begin{figure}[h!]
    \begin{center}
		\caption{Dependence of $\mathbb{E}\left[ x_{11,t_1} \mid  x_{11,t_0}>x_* \right]$ on \(t_1\).}
    \label{fig:t1_vs_output_plot}
		\includegraphics[width=0.6\textwidth]{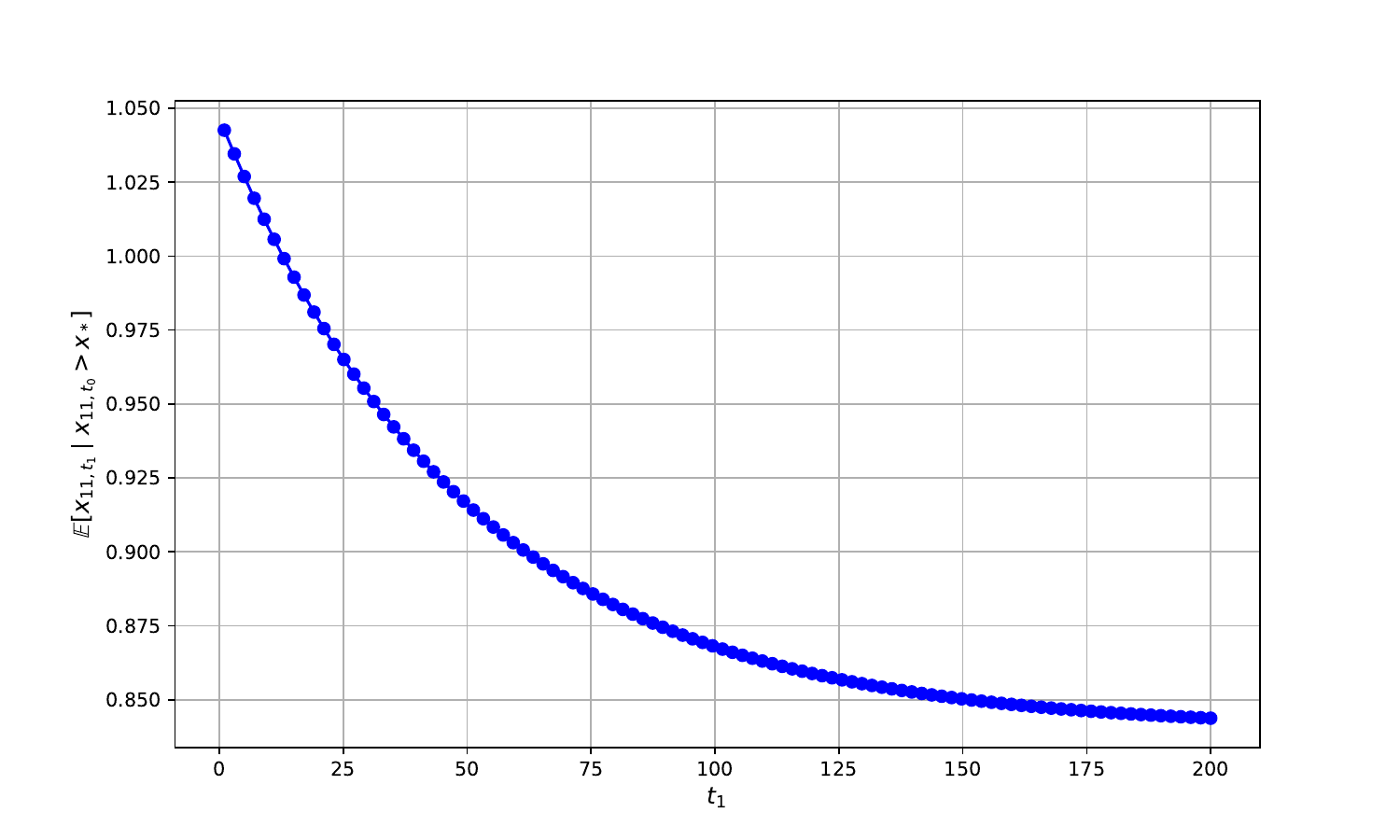}
    \end{center}
		Note: Two dates expected shortfall depending on $t_1$. We calculate $\mathbb{E}\left[ x_{11,t_1} \mid  x_{11,t_0}>x_* \right]$, where we fix $t_0 = 1$ and vary $t_1$ in the range $[1, 200]$.
\end{figure}

\clearpage

\setcounter{page}{1}

\begin{center}
 Online supplementary materials to:\\

\Large Wishart conditional tail risk measures: An analytic approach\\

\normalsize by\\ 

\large Jos\'e Da Fonseca and Patrick Wong
\end{center}
\normalsize

\paragraph{A comparison with \citet{IgnatievaLandsman2019} and \citet{IgnatievaLandsman2025}:} We provide complete details, using the methodology presented in this paper, on the implementation of the expected shortfall, the tail variance and the tail skewness for a univariate random variable following a Generalised Hyperbolic (GH) distribution. For the expected shortfall, it is given by \eqref{eq:TCE_IL_1} and corresponds to \citet[Eq. (3.4)]{IgnatievaLandsman2019}. As in \citet{IgnatievaLandsman2019}, we derive the expression for the characteristic function of the univariate GH distribution from the moment generating function \eqref{eq:GH_mgf}, and it corresponds to \citet[Eq.~(2.6)]{IgnatievaLandsman2019}, where we correct for a missing square root in the modified Bessel function of the third kind that appears in the denominator. Following \citet{IgnatievaLandsman2019}, see \citet{Paolella2007} for further details.\\

Given a scalar variable \(Y\sim GH_1(\lambda, \chi, \psi,\mu, \sigma^2,\gamma)\) with \(\lambda, \chi, \psi,\mu, \sigma^2,\gamma\) scalar values and \(\chi>0\) and \(\psi>0\), the univariate GH distribution with density \eqref{eq:GH_density}. Then the moment generating function of $Y$, given by \eqref{eq:GH_mgf}, is of scalar argument and simplifies to
\begin{align}
    \Phi_{Y}(v) = e^{v \mu} \left(\frac{\psi}{\psi - 2 v \gamma - v^2 \sigma^2}\right)^{\lambda / 2}  \frac{K_\lambda\left(\sqrt{\chi \left(\psi - 2 v\gamma - v^2  \sigma^2\right)}\right)}
    {K_\lambda\left(\sqrt{\chi \psi}\right)}. \label{eq:GH_mgf_scalar}
\end{align}
The tail conditional expectation \(\TCE_Y(y_*)\) is equal to \eqref{eq:ConditionalMomentYp} with \(p=1\), that is
\begin{align}
\mathbb{E}[Y | Y>y_*]=\frac{\mathbb{E}[Y \mathbf{1}_{\lbrace Y>y_* \rbrace}]}{\mathbb{E}[\mathbf{1}_{\lbrace Y>y_*\rbrace}]}, \label{eq:expected_shortfall_appendix}
\end{align}
and according to Proposition \ref{prop:ExpectedShortfallFourierTransform}, the expectation in the numerator is given by
\begin{align}
\mathbb{E}[Y \mathbf{1}_{\lbrace Y>y_* \rbrace}]=\frac{1}{\pi}\int_{0}^{+\infty} \Re\left( \left(\frac{ y_*}{(\alpha_1-\mathrm{i}u_1)}+\frac{1}{(\alpha_1-\mathrm{i}u_1)^{2}} \right) e^{-(\alpha_1 -\mathrm{i}u_1)y_* }\Phi_Y(\alpha_1-\mathrm{i} u_1)  \right)du_1,\label{eq:MomentY1_numerator}
\end{align}
while the denominator is 
\begin{align}
\mathbb{E}[\mathbf{1}_{\lbrace Y>y_* \rbrace}]=\frac{1}{\pi}\int_{0}^{+\infty} \Re\left( \frac{ e^{-(\alpha_1 -\mathrm{i}u_1)y_* }\Phi_Y(\alpha_1-\mathrm{i} u_1)}{(\alpha_1-\mathrm{i}u_1)} \right)du_1. \label{eq:MomentY1_denominator}
\end{align}
It remains to check the integrability of these functions. The second integrand in \eqref{eq:MomentY1_numerator} leads to
\begin{align}
\left| \int_{0}^{+\infty} \Re\left( \frac{e^{-(\alpha_1 -\mathrm{i}u_1)y_* }\Phi_Y(\alpha_1-\mathrm{i} u_1)}{(\alpha_1-\mathrm{i}u_1)^{2}}\right)du_1 \right|\leq c \int_{0}^{+\infty} \frac{|\Phi_Y(\alpha_1-\mathrm{i} u_1)|}{|(\alpha_1-\mathrm{i}u_1)^{2}|}du_1, 
\end{align}
with \(c>0\) and as \(|\Phi_Y(\alpha_1-\mathrm{i} u_1)| \leq |\mathbb{E}[e^{(\alpha_1-\mathrm{i} u_1)Y}]| \leq |\mathbb{E}[e^{\alpha_1Y}]| \) is finite for \(\alpha_1\) sufficiently small since \(\psi>0\), we conclude the integral is well defined. Note that \(\alpha_1\) can have any sign. The integral \eqref{eq:MomentY1_denominator} and the first integrand of \eqref{eq:MomentY1_numerator} are also well defined since we have
\begin{align}
\left|\int_{0}^{+\infty} \Re\left( \frac{ e^{-(\alpha_1 -\mathrm{i}u_1)y_* }\Phi_Y(\alpha_1-\mathrm{i} u_1)}{(\alpha_1-\mathrm{i}u_1)} \right)du_1\right|\leq c \int_{0}^{+\infty} \frac{|\Phi_Y(\alpha_1-\mathrm{i} u_1)|}{|(\alpha_1-\mathrm{i}u_1)|} du_1,\label{eq:upper_bound1}
\end{align}
with \(c>0\) (a generic constant that may change from one equation to the other), and using \citet[(10.25.3)]{DLMF} we have for \(u_1\) large (and positive)
\begin{align}
|\Phi_Y(\alpha_1-\mathrm{i} u_1)|\leq c \frac{e^{- \sqrt{\chi(\psi - 2\alpha_1 \gamma + u_1^2\sigma^2) + \mathrm{i}\chi(2u_1\gamma + 2\alpha_1 u_1\sigma^2)}}}{|(\psi - 2\alpha_1 \gamma + u_1^2\sigma^2) + \mathrm{i}(2u_1\gamma + 2\alpha_1 u_1\sigma^2)|^{\lambda/2 + 1/4}},
\end{align}
with \(c>0\). For \(u_1\) large we get, since \(\chi>0\), that there exists \(c>0\) such that
\begin{align}
|\Phi_Y(\alpha_1-\mathrm{i} u_1)|\leq c \frac{e^{- \sqrt{\chi} u_1\sigma}}{u_1^{\lambda + 1/2}},
\end{align}
and this upper bound combined with \eqref{eq:upper_bound1} implies the integral is well defined. As result, \eqref{eq:MomentY1_numerator} and \eqref{eq:MomentY1_denominator} are well defined and therefore the expected shortfall \eqref{eq:expected_shortfall_appendix} is also well defined.\\

For the tail variance, when \(p=2\) in \eqref{eq:ConditionalMomentYp} leads to 

\begin{align}
\mathbb{E}[Y^2 \mathbf{1}_{\lbrace Y>y_* \rbrace}]=\frac{1}{\pi}\int_{0}^{+\infty} \Re\left( \left(\frac{ y_*^2}{(\alpha_1-\mathrm{i}u_1)} + \frac{ 2y_*}{(\alpha_1-\mathrm{i}u_1)^2}+\frac{2}{(\alpha_1-\mathrm{i}u_1)^{3}} \right) e^{-(\alpha_1 -\mathrm{i}u_1)y_* }\Phi_Y(\alpha_1-\mathrm{i} u_1)  \right)du_1, \label{eq:MomentY2_numerator}
\end{align}
 
and, as previously mentioned, the \(\TV_Y(y_*)\) of \eqref{eq:TailVariance} is known and given by
\begin{align}
 \TV_Y(y_*) =\frac{ \mathbb{E}[Y^2 \mid Y > y_*]}{\mathbb{E}[\mathbf{1}_{\lbrace Y > y_*\rbrace }]} - (\mathbb{E}[Y \mid Y > y_*])^2. \label{eq:TV_example_mgf}
\end{align}
The integral above can be decomposed as a sum of three integrals, the first two being well defined according to the verifications already done, and the last one is also well defined. As a result, the tail variance risk representation based on the MGF is well defined.\\

In \citet{IgnatievaLandsman2025}, the following explicit expression is derived for the tail variance risk measure when \(Y\sim GH_1(\lambda, \chi, \psi,\mu, \sigma^2,\gamma)\):
\begin{align}
\TV_Y(y_*)	&=  \frac{\sigma^2 k_\lambda }{1-q}  \bar{F}_{\mathrm{GH}_1}(y_*; \lambda + 1, \chi, \psi, \mu, \sigma^2, \gamma)\left( 1 +(y_*-\mu)\frac{ f_{\mathrm{GH}_1}(y_*; \lambda + 1, \chi, \psi, \mu, \sigma^2, \gamma)}{ \bar{F}_{\mathrm{GH}_1}(y_*; \lambda + 1, \chi, \psi, \mu, \sigma^2, \gamma)} \right) \nonumber \\
						&+  \frac{\tilde{k}_\lambda \gamma }{1-q}  \bar{F}_{\mathrm{GH}_1}(y_*; \lambda + 2, \chi, \psi, \mu, \sigma^2, \gamma)\left( \gamma +\sigma^2\frac{ f_{\mathrm{GH}_1}(y_*; \lambda + 2, \chi, \psi, \mu, \sigma^2, \gamma)}{ \bar{F}_{\mathrm{GH}_1}(y_*; \lambda + 2, \chi, \psi, \mu, \sigma^2, \gamma)} \right) \nonumber \\
						&- \frac{k_\lambda^2}{(1-q)^2} \bar{F}_{\mathrm{GH}_1}(y_*; \lambda + 1, \chi, \psi, \mu, \sigma^2, \gamma)^2\left( \gamma +\sigma^2\frac{ f_{\mathrm{GH}_1}(y_*; \lambda + 1, \chi, \psi, \mu, \sigma^2, \gamma)}{ \bar{F}_{\mathrm{GH}_1}(y_*; \lambda + 1, \chi, \psi, \mu, \sigma^2, \gamma)} \right)^2   , \label{eq:TV_IL_1}
\end{align}
with \(k_\lambda\) as in \eqref{eq:k_lambda} and 
\begin{align*}
\tilde{k}_\lambda = \frac{\chi}{\psi} \frac{K_{\lambda+2}(\sqrt{\chi \psi})}{K_\lambda(\sqrt{\chi \psi})}. 
\end{align*}

\begin{remark}
Note that there is a typo in \citet[Eq.~(3.6)]{IgnatievaLandsman2025}. In the second line of this equation, the term \( \bar{F}_{\mathrm{GH}_1}(x_q;\mu,\sigma^2, g, \gamma, l,\tilde{\lambda}_2,\chi,\psi)\) should be replaced by \( \bar{F}_{\mathrm{GH}_1}(x_q;\mu,\sigma^2, g, \gamma, l,\tilde{\lambda}_4,\chi,\psi)\) (in the notations of \citet{IgnatievaLandsman2025}).
\end{remark}

Of course, \eqref{eq:TV_example_mgf} computed using the MGF and \eqref{eq:TV_IL_1} should yield the same numerical results.\\

Lastly, for the tail skewness, when \(p=3\) in \eqref{eq:ConditionalMomentYp}, the simplicity of the method proposed in our work allows us to compute
\begin{align}
\mathbb{E}[Y^3 \mathbf{1}_{\lbrace Y>y_* \rbrace}]&=\frac{1}{\pi}\int_{0}^{+\infty} \Re\left( \left(\frac{ y_*^3}{(\alpha_1-\mathrm{i}u_1)} + \frac{ 3y_*^2}{(\alpha_1-\mathrm{i}u_1)^2}+\frac{6y_*}{(\alpha_1-\mathrm{i}u_1)^{3}}\right) e^{-(\alpha_1 -\mathrm{i}u_1)y_* }\Phi_Y(\alpha_1-\mathrm{i} u_1)  \right)du_1 \nonumber \\
&+\frac{1}{\pi}\int_{0}^{+\infty} \Re\left( \frac{ 6}{(\alpha_1-\mathrm{i}u_1)^4} e^{-(\alpha_1 -\mathrm{i}u_1)y_* }\Phi_Y(\alpha_1-\mathrm{i} u_1)  \right)du_1, \label{eq:MomentY3_numerator}
\end{align}
with the first integral above being well-defined, since \eqref{eq:MomentY2_numerator} is, while the last integral is also well-defined. Combining this expectation with the tail conditional variance \eqref{eq:TV_example_mgf} and the expression \eqref{eq:tail_skew} gives us the tail skewness, which is yet to be computed in the case of \citet{IgnatievaLandsman2025}.\\

Let us verify that the tail conditional expectation and tail variance values obtained using the approach based on the MGF match those reported by \citet{IgnatievaLandsman2019} and \citet{IgnatievaLandsman2025}. For the tail skewness, we provide its value, which could be used to validate a formula based on the density once such a formula becomes available.

\begin{table}[h]
\begin{center}
\caption{Risk measures.}
\label{tab:risk_measures_comparison}
\[
\begin{array}{lcc}
\textrm{Risk measure} & \textrm{Value density}& \textrm{Value mgf} \\ \hline
\TCE_Y(y_*)						& 2.4649 								& 2.4649 		\\
\TV_Y(y_*)       			& 0.8749 								& 0.8749		\\
\TS_Y(y_*)       			&  											& 2.6444 		\\ \hline
\end{array}
\]

\end{center}
{\footnotesize \textit{Note.} This table contains the risk measures \(\TCE, \TV, \TS\) for the GH distribution computed using the density and given by \eqref{eq:TCE_IL_1} (or \citet[Eq.~(2.6)]{IgnatievaLandsman2019}), and \eqref{eq:TV_IL_1} (or \citet[Eq.~(3.6)]{IgnatievaLandsman2025}), as well as using the moment generating function \eqref{eq:expected_shortfall_appendix} and \eqref{eq:TV_example_mgf}. For the tail skewness, it is computed using \eqref{eq:tail_skew}, which is based on  \eqref{eq:MomentY3_numerator}, and relies on the moment generating function. The parameters are those for the aggregate portfolio \citet[Table 1, Panel F]{IgnatievaLandsman2019} while $q=0.95$ and $y_*$ is the $q$ level quantile.}
\end{table}

\begin{remark}
Note that our value for \(\TCE\) does not correspond to the reported value (\textit{i.e.}, \(2.219615\)) in \citet[Table 4, Panel F, column GH]{IgnatievaLandsman2019}, which implies that there is an implementation error in that paper.
\end{remark}


\paragraph{A simple implementation of the method of moments estimator:} This part records the method of moments formulas used to estimate $(\varsigma_{\infty},\, \beta)$ in Corollary~\ref{coro:MGFInfinity}, the parameters of the limiting distribution of $x_t$ in the \citet{Bru1991} case (\textit{i.e.,} when $\omega = \beta \sigma^2$). See \citet[Section 5]{GourierouxJasiakSufana2009} and \citet[Prop.~2 and Prop.~5]{GouerierouxLu2025} for derivations.\\

At the population level (with $i\neq j$), the identities are
\begin{align}
    \beta &= \frac{2\,\big(\mathbb{E}[x_{ii,\infty}]\big)^2}{\operatorname{Var}[x_{ii,\infty}]}, \\
    \varsigma_{ii,\infty} &= \frac{\operatorname{Var}[x_{ii,\infty}]}{2\,\mathbb{E}[x_{ii,\infty}]}, \\
    \varsigma_{ij,\infty}^{2} &= \frac{\operatorname{Cov}[x_{ii,\infty},x_{jj,\infty}]}{2\beta},
\end{align}
where $\beta$ does not depend on the index $i$. As result, there are several estimator for \(\beta\) and for \(\varsigma_{ij,\infty}^{2}\).\\

\citet{GouerierouxLu2025} suggest the following estimation procedure. Let $\widehat{\mu}_i$, $\widehat{\eta_i^{2}}$, and $\widehat{\gamma}_{ij}$ denote the sample mean of $x_{ii,t}$, the sample variance of $x_{ii,t}$, and the sample covariance of $(x_{ii,t},x_{jj,t})$, respectively. The plug-in estimators are
\begin{equation}
    \widehat{\beta}=\frac{2}{n}\sum_{i=1}^{n}\frac{(\widehat{\mu}_i)^{2}}{\widehat{\eta_i^{2}}}, 
    \qquad 
    \widehat{\varsigma_{ii,\infty}}=\frac{\widehat{\eta_i^{2}}}{2\,\widehat{\mu}_i}, 
    \qquad 
    \widehat{\varsigma_{ii,\infty}^{2}}=\frac{\widehat{\gamma}_{ij}}{2\widehat{\beta}}.
    \label{eq:MoM_estimators}
\end{equation}

Note that we also have

\begin{align}
    \mathbb{E}[x_{ij,\infty}]  =  \beta  \varsigma_{ij,\infty}. \label{eq:varsigma_ij_infty}
\end{align}

\begin{remark}
Relative to \citet{GouerierouxLu2025}, there is a missing factor of $2$ in the formulas for $\beta$, $\varsigma_{ii,\infty}$ and $\varsigma_{ij,\infty}$. Further to this, the $\varsigma_{ij,\infty}$ in \citet[Eq. (7)]{GouerierouxLu2025} is off by a factor of $\beta$.
\end{remark}

\end{document}